\PassOptionsToPackage{usenames,dvipsnames}{xcolor}
\documentclass[11pt]{article}

\usepackage[margin=1in]{geometry}
\usepackage[T1]{fontenc}
\usepackage{charter}
\usepackage[english]{babel}

\usepackage{amsmath, amssymb, amstext, amsthm}
\usepackage{mathtools}  
\usepackage{thmtools, thm-restate}
\usepackage{mathrsfs}
\usepackage{nicefrac}
\usepackage{mleftright} 

\usepackage{setspace}
\usepackage[inline]{enumitem}
\usepackage{multicol}
\usepackage{csquotes}

\usepackage{graphicx}
\usepackage{tcolorbox}
\usepackage{todonotes}
\usepackage{comment}

\usepackage{hyperref}
\usepackage{cleveref}

\addto\extrasenglish{%
}

\makeatletter
\def\algbackskip{\hskip-\ALG@thistlm}
\makeatother

\title{When agents choose bundles autonomously: guarantees beyond discrepancy}

\author{
    Sushmita Gupta 
    \thanks{The Institute of Mathematical Sciences, HBNI, Chennai. {\tt sushmitagupta@imsc.res.in}} \and 
    Pallavi Jain 
    \thanks{Indian Institute of Technology Jodhpur, Jodhpur. {\tt pallavi@iitj.ac.in}} \and 
    Sanjay Seetharaman
    \thanks{The Institute of Mathematical Sciences, HBNI, Chennai. {\tt sanjays@imsc.res.in}} \and
    Meirav Zehavi
    \thanks{Ben-Gurion University of the Negev, Beer-Sheva. {\tt meiravze@bgu.ac.il}}
}

\def\version{compressed} 

\def\algostyle{algorithm2e} 

\newcommand{\etal}{et al.\xspace}

\DeclareMathSymbol{\qm}{\mathalpha}{operators}{"3F}
\DeclareMathAlphabet{\mathbbold}{U}{bbold}{m}{n}

\newcommand{\ceil}[1]{\ensuremath{\left\lceil #1 \right\rceil}}

\newcommand\bigoh{\mathcal{O}}

\newcommand{\sse}{\subseteq\!}

\newcommand{\M}[1]{\ensuremath{M[#1]}}

\newcommand{\org}[1]{\hide{#1}}
\newcommand{\ma}[1]{\hide{#1}}
\newcommand{\Ma}[1]{#1}

\newcommand{\later}[1]{\todo[inline, color=yellow!40]{ }}

\newcommand{\NPC}{{\sf NP}-{\sf complete}\xspace}

\newcommand{\hide}[1]{{}}

\renewcommand{\part}{part\xspace}
\newcommand{\parts}{parts\xspace}

\newcommand{\seq}[1]{\ensuremath{\langle{1, \ldots, #1\rangle}}}

\newcommand{\prop}{\mathsf{PROP}}

\newcommand{\disc}{\mathrm{disc}}

\newcommand{\change}{\mathrm{change}}

\usepackage{soul}

\newcommand{\rr}{Round-Robin\xspace}

\newcommand{\app}{$\dagger$}

\usepackage{ifthen}

\ifthenelse{\equal{\algostyle}{algorithm2e}}{
    \usepackage[noend,ruled,linesnumbered]{algorithm2e} 
    
    \SetAlFnt{\small}
    \SetAlCapFnt{\small}
    \SetAlCapNameFnt{\small}
    \SetAlCapHSkip{0pt}
    \IncMargin{-\parindent}
    \crefname{algocf}{alg.}{algs.}
    \Crefname{algocf}{Algorithm}{Algorithms}
}{}

\ifthenelse{\equal{\algostyle}{algorithmic}}{
    \usepackage{algorithm}
    \usepackage{algpseudocode}

}{}

\usepackage{refcount}

\declaretheorem[numberwithin=section]{theorem}
\declaretheorem[numberlike=theorem,style=definition]{definition}
\declaretheorem[numberlike=theorem]{claim,observation,corollary,proposition,lemma}
\declaretheorem[numberlike=theorem,style=remark]{remark}

\Crefname{assumption}{Assumption}{Assumptions}
\declaretheorem[numberwithin=theorem,name=Claim]{claim-inside-theorem}
\Crefname{claim-inside-theorem}{Claim}{Claims}
\declaretheorem[numberlike=claim-inside-theorem,name=Lemma]{lemma-inside-theorem}
\Crefname{lemma-inside-theorem}{Lemma}{Lemmas}
\declaretheorem[numberlike=claim-inside-theorem,name=Observation]{observation-inside-theorem}
\declaretheorem[numberwithin=lemma,name=Claim]{claim-inside-lemma}
\Crefname{claim-inside-lemma}{Claim}{Claims}
\declaretheorem[numberlike=claim-inside-lemma,name=Corollary]{corollary-inside-lemma}
\declaretheorem[numberwithin=lemma,name=Lemma]{lemma-inside-lemma}
\declaretheorem[numberlike=claim-inside-lemma,name=Observation]{observation-inside-lemma}
\declaretheorem[numberwithin=lemma-inside-theorem,name=Claim]{claim-inside-lemma-inside-theorem}
\declaretheorem[numberwithin=lemma-inside-theorem,name=Lemma]{lemma-inside-lemma-inside-theorem}
\declaretheorem[numberwithin=lemma-inside-theorem,name=Observation]{observation-inside-lemma-inside-theorem}

\allowdisplaybreaks

\begin{document}
\date{}
\begin{sloppypar}

\maketitle

\begin{abstract}
We consider the fair division of indivisible items among $n$ agents with additive non-negative normalized valuations, with the goal of obtaining high value guarantees, that is, close to the proportional share for each agent. 

We prove that partitions where \emph{every} part yields high value for each agent are asymptotically limited by a discrepancy barrier of $\Theta(\sqrt{n})$.
Guided by this, our main objective is to overcome this barrier and achieve stronger individual guarantees for each agent in polynomial time. 

Towards this, we are able to exhibit an exponential improvement over the discrepancy barrier. 
In particular, we can create partitions on-the-go such that when agents arrive sequentially (representing a previously-agreed priority order) 
and pick a part autonomously and rationally (i.e., one of highest value), then each is guaranteed a part of value at least $\mathsf{PROP} - \mathcal{O}{(\log n)}$.
Moreover, we are able to show even better guarantees for three restricted valuation classes such as those defined by: a common ordering on items, a bound on the multiplicity of values, and a hypergraph with a bound on the \emph{influence} of any agent.
Specifically, we study instances where: (1) the agents are ``close'' to unanimity in their relative valuation of the items. This is a generalization of the ordered additive setting, and by analyzing the structure of the swaps--adjacent, linearly-separable, or laminar--we are able to guarantee each agent a value close to their proportional share; 
(2) the valuation functions do not assign the same positive value to more than $t$ items; and
(3) the valuation functions respect a hypergraph, a setting introduced by Christodoulou et al. [EC'23], where agents are vertices and items are hyperedges. While the size of the hyperedges, and neighborhoods can be arbitrary, and the valuation function can be asymmetric between agents that share an edge, the influence of any agent $a$, defined as the number of its neighbors who value at least one item positively that $a$ also values positively, is bounded. 
\end{abstract}

\thispagestyle{empty}
\newpage

\setcounter{tocdepth}{1}
\tableofcontents
\thispagestyle{empty}
\newpage
\setcounter{page}{1}

\section{Introduction}\label{s:intro}
\label{sec:intro}

Fair division of indivisible items is a central problem in computational social choice. 
A core difficulty is reconciling strong fairness guarantees with the combinatorial constraints imposed by indivisibility. 
Among the many fairness notions studied, proportionality is one of the most fundamental: each agent should receive a bundle valued at least a $1/n$ fraction of their total value for the entire set of items. 
This threshold is called the agent's \emph{proportional share}.
While exact proportionality is often impossible even under additive valuations, it remains a natural benchmark for evaluating the quality of allocations.

In this work, we focus on the standard setting of fair division with additive, non-negative valuation functions.
Formally, we consider an instance with a set of $n$ agents $[n] = \{1, \dots, n\}$ and a set of $m$ indivisible items $[m] = \{1, \dots, m\}$.
Each agent $i$ has a valuation function $v_i$ that assigns a non-negative value to every subset of items. 
We assume valuations are additive: $v_i(S) = \sum_{g \in S} v_i(g)$ for any subset $S \subseteq [m]$.
To simplify our presentation, we adopt the standard normalization where the maximum value of any single item is at most $1$ (i.e., $\max_{g \in [m]} v_i(g) \le 1$ for all $i$).
We denote the proportional share of agent $i$ by $\prop_i \coloneq v_i([m])/n$.

In our model, the allocation is determined by a two-stage process. 
First, the items are divided into an $n$-partition $M = \{\M{1}, \dots, \M{n}\}$. 
Then, the agents act \emph{autonomously} to select a \part from this partition. 
We assume agents are \emph{rational}: each agent secures a \part that maximizes their valuation among the available options. 
\hide{Equivalently, the algorithm fixes an $n$-partition in advance, and agents sequentially select and remove \parts.}
Our primary goal is to find a partition that guarantees every agent a \part with value as close to their proportional share as possible. 

We note that, in general, determining if there exists an allocation giving each agent their proportional share is a well-known \NPC problem~\cite{demko1988equitable}, and has been studied extensively~\cite{DBLP:conf/sigecom/ConitzerF017,DBLP:conf/aaai/BaklanovGGS21,DBLP:conf/sagt/BismuthBS24}. 
Furthermore, standard algorithms like \emph{\rr} or \emph{Envy-Cycle Elimination} can generate an allocation where each agent $i$ is \emph{assigned} a \part of value at least $\prop_i - 1$. 
The critical difference is that these mechanisms assign \parts directly; unlike in our model, the agents do not autonomously select the \part they value most.

Next, we discuss our work in the broader context of the fair division literature, particularly discrepancy theory, while formally introducing our terminology as it becomes relevant.

\subsection{Our contribution}\label{subsec:contribution}
Our work in this article is three-fold: 
\begin{enumerate*}
    \item we exhibit families of instances establishing the impossibility of obtaining strong \emph{universal guarantees}, i.e., partitions where every \part yields high value for every agent; 
    \item we formalize a model of fair division in which rational agents are presented with a partition and autonomously select the \part that maximizes their value; and 
    \item we prove several algorithmic results in this model that provide significant improvements to the \emph{individual guarantee} (the value of the \part actually obtained by an agent) over existing methodologies, all attainable in polynomial time.
\end{enumerate*}


We begin by establishing an impossibility result that goes well beyond the guarantees discrepancy theory can provide for our model.  
While we defer formal details on discrepancy theory to Appendix~\ref{sec:prelim-discrepancy}, the core problem can be summarized as follows: given a set of vectors in $\mathbb{R}^n$, the goal is to partition them into $n$ sets such that the vector sum of each set is ``close'' to $\frac{1}{n}$ fraction of the total sum of all vectors.
Naturally, results of this kind are relevant to our model.
Our impossibility result, specifically \Cref{thm:lower-bound}, may be of independent interest to readers. 

\subsubsection*{Lower bound showing impossibility in general setting}

Based on the results in discrepancy theory by Doerr and Srivastav~\cite{DBLP:journals/cpc/DoerrS03} and Manurangsi and Suksumpong~\cite{DBLP:journals/tcs/ManurangsiS22} it is known that for any instance with $n$ agents, there exists an $n$-partition $(B_1, \ldots, B_n)$ such that the value of every part for every agent lies within a small window  around the proportional share. 
Specifically:
\[
\prop_i - \bigoh(\sqrt{n}) \leq v_i(B_j) \leq \prop_i + \bigoh(\sqrt{n}) \quad \forall i, j \in [n].
\]

On the impossibility side, Manurangsi and Meka~\cite{doi:10.1137/1.9781611978964.20Manurangsi26} have shown that this bound is tight.
Specifically, for any $n$ that is a power of $2$, there exists an instance with $n$ agents where, for \emph{any} $n$-partition $(B_1, \dots, B_n)$, there exists an agent $a$ and a part $B_j$ such that the deviation is significant:
\[
|v_a(B_j) - \prop_a| \ge \sqrt{(n-1)}/16 = \Omega(\sqrt{n}).
\]

However, we observe that this result does not necessarily imply an impossibility for our specific goal. 
Discrepancy is a two-sided metric (which deals with both deviations above and below the mean), whereas our objective is to maximize the guarantee \emph{from below} (ensuring agents receive \emph{at least} a certain value).
Thus, the impossibility of low discrepancy does not preclude the existence of allocations where every part has a high minimum value. 
For instance, suppose that an instance admits a partition where an agent $a$ values the \parts as:
\[
    (\prop_a + (n-1), \prop_a - 1, \dots, \prop_a - 1).
\]
This partition witnesses the discrepancy impossibility result (due to the large positive deviation in the first \part), yet every \part has value at least $\prop_a-1$, making it very good for our objective.

In light of the impossibility shown by Manurangsi and Meka~\cite{doi:10.1137/1.9781611978964.20Manurangsi26} and the intuition from the example above, a natural strategy is to relax the guarantee from above and focus solely on improving the guarantee from below. 
However, we show that the discrepancy-based universal guarantees are in fact asymptotically tight even from below. 
Specifically, we demonstrate the existence of an instance where, regardless of the chosen partition, there is always a \part that fails to meet the guarantee from below for some agent. 
Formally, we prove the following:

\begin{restatable}{theorem}{lowerbound}
    \label{thm:lower-bound}
    For any integer $n \ge 4$ that is a power of two, there exists a fair division instance $([m], v_1, \dots, v_n)$ with binary valuations such that the following holds:
    For any $n$-partition $(P_1, \dots, P_n)$ of the items, there exists an agent $a$ and a \part $P_j$ such that 
    \[
    v_a(P_j) \le \prop_a - \sqrt{n/32}.
    \]
\end{restatable}

\hide{We conclude our discussion of the impossibility results by noting that discrepancy-based {\it universal guarantees}--where every part gives a large value for each agent--that are asymptotically better than $\prop_i - \sqrt{n}$ are unlikely to exist due to the fact that the discrepancy guarantee of $\prop_i \pm \sqrt{n}$ cannot be improved, as shown recently by Manurangsi and Meka~\cite{doi:10.1137/1.9781611978964.20Manurangsi26}. Consequently, forsaking the guarantee from above and just focusing on improving the guarantee from below is the natural next step. However, that possibility also seems to run into trouble in light of \Cref{thm:lower-bound} which proves that the discrepancy-based guarantees are asymptotically tight from below. \ma{this feels repetitiive given what has been said in last 2 paras. (Moved above)}
}

This theorem establishes that one cannot guarantee that all \parts are simultaneously valuable for all agents. Consequently, our goal shifts to ensuring that each agent has access to---and ultimately receives---at least one high-value part.


We assume agents are primarily self-interested: if allowed to choose autonomously, they will act \emph{rationally} to select the best available \part. 
Since multiple agents may make the same choice, instead of breaking ties in an ad-hoc fashion, we assume a fixed arrival order of agents. 
Specifically, we consider the following \emph{sequential arrival model}: 
\begin{quote}
    An $n$-partition $P$ is fixed and presented to the agents. 
    Agents arrive sequentially according to a fixed order, pick the highest-valued part remaining in $P$, and depart.
\end{quote}

It is straightforward to see that a universal guarantee provides a pessimistic guarantee in this model. 
If a partition ensures every part is good for everyone, then regardless of what previous agents select, the remaining parts will satisfy the current picking agent. 
However, we know from \Cref{thm:lower-bound} that such universal guarantees are inherently limited by the $\Theta(\sqrt{n})$ barrier.

The core challenge is to determine if the \emph{power of autonomy} allows us to surpass this barrier and achieve a guarantee of $\prop - o(\sqrt{n})$. 
This is not immediate because, for an arbitrary partition, rational choice can be destructive: early agents might pick the specific parts that are valuable to later agents. 
As discussed in \Cref{rem:need-for-rebundling}, this can leave later agents with parts of \emph{zero} value, creating an unfair outcome. 
Our contribution is to construct partitions specifically designed to avoid such possibilities.

Towards our goal of obtaining better guarantees, two natural research directions emerge. 
The first is to demonstrate the existence of a partitioning strategy such that, for any given ordering of agents, the best available \part for each agent yields a value of at least $\prop - o(\sqrt{n})$. 
The alternative is to prove an impossibility result: showing that there exists an instance and an arrival order such that, for any partitioning strategy, some agent is forced to take a part with value strictly less than $\prop - \omega(1)$.
Addressing either one\hide{a positive existence result or an impossibility result} would be worthwhile.

In this article, we take the first approach. 
We use \rr as a guiding principle and modify it appropriately to achieve polynomial-time algorithmic guarantees that offer an exponential improvement over the discrepancy barrier.
To prevent early agents from depleting the high-value parts needed by later agents, our algorithm employs a strategy of \emph{rebundling}. 
Rather than fixing a single partition at the start, the mechanism adapts the \parts of remaining items as the process unfolds.
We concede that this ``dynamic'' rebundling moves us away from using the same fixed partition for all the agents. 
We believe this trade-off is justified given that the result we obtain is an exponential improvement in the individual guarantee for every agent---specifically, $\prop_i - \bigoh(\log n)$, for each agent $i \in [n]$---that too via a polynomial-time algorithm.

\hide{Furthermore, we show that if the mechanism is allowed to determine the arrival sequence, there exists a specific ordering (computable in polynomial time), which we call the \emph{fair arrival order}, that ensures these strong guarantees are met. \ma{ss: feel this line is out of place. the fair arrival order is for removing the assumption on prop share of the arriving agents}}



\paragraph{Algorithmic guarantees.} 
We classify our positive results into two categories: 
(I) \textit{Static partitioning}, where a single partition is fixed at the start, and all agents choose from this common menu (excluding parts picked by predecessors); and 
(II) \textit{Dynamic partitioning}, where the mechanism rebundles the remaining items between agent arrivals.


\subsubsection*{Bounded distance from unanimity}

A fair division instance with additive valuations is said to be \emph{ordered additive} if $v_i(g_1) \ge \dots \ge v_i(g_m)$ for all $i \in [n]$~\cite{DBLP:journals/aamas/BouveretL16}. 
Within the static framework (Category I), we focus on instances that are ``close'' to being ordered additive. 
This is particularly relevant when agents are \emph{nearly unanimous} in their item rankings. 
Formally, we assume that the agents' valuations lie at a small distance from a \emph{master list}, defined as a fixed permutation $\pi$ of the item set $[m]$. 
A valuation function $v$ is said to be \emph{derived} from a master list $\pi$ if, for any two items $g_1, g_2$, the condition $v(g_1) > v(g_2)$ implies that $g_1$ precedes $g_2$ in $\pi$. 
Observe that in an ordered-additive instance, all agent valuations are derived from the identity permutation $\langle 1, \dots, m \rangle$; effectively, the master list $\langle 1, \dots, m \rangle$ is at distance zero from every agent. 
Thus, such instances form the baseline of our analysis.

To model deviations from this baseline, we measure the distance between permutations based on the number and nature of \emph{swaps}. 
We analyze a variety of swap types, ranging from \emph{arbitrary} and \emph{adjacent} swaps (restricted to consecutive positions) to structured families defined by the relationships between them, such as \emph{linearly-separable} and \emph{laminar} swaps (defined formally in \Cref{subsec:linearly-separable-swaps}).
These structural distinctions allow us to analyze finer separations in the behavior of the underlying valuations and obtain stronger guarantees. 
\Cref{fig:types-of-swaps} illustrates the different swap types we consider.

\begin{figure}
    \centering
    \includegraphics[width=\linewidth]{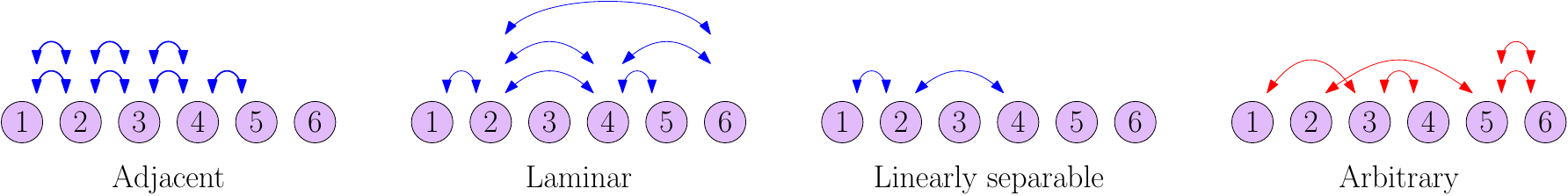}
    \caption{An illustration of the different types of swaps we consider.}
    \label{fig:types-of-swaps}
\end{figure}

In the following result, we use $\sigma_a$ to denote a fixed preference ordering of agent $a$ over the items.

\begin{restatable}{theorem}{ordaddrelaxationsguarantee}
\label{thm:ordered-additive-relaxation-guarantees}
Given a fair division instance and a master list $\pi$, let $k_a$ denote the number of swaps required to transform $\sigma_a$ to $\pi$ under the relevant swap metric.
There exists a polynomial-time algorithm that, depending on the nature of the instance, generates a partition $\mathcal{M}=\{\M{1}, \ldots, \M{n}\}$ satisfying the following for each agent $a$ and any part $\M{j} \in \mathcal{M}$:
\begin{description}
    \item[Ordered additive:] $v_a(\M{j}) \ge \prop_a - 1$.
    \item[Linearly separable swaps:] $v_a(\M{j}) \ge \prop_a - 2$.
    \item[Adjacent swaps:] $v_a(\M{j}) \ge \prop_a - 1 - \sqrt{2k_a}$.
    \item[Arbitrary swaps:] $v_a(\M{j}) \ge \prop_a - 1 - k_a$.
    \item[Laminar swaps:] $v_a(\M{j}) \ge \prop_a - 1 - \text{depth}(S^a)$, where $S^a$ is a laminar family of swaps required to transform $\sigma_a$ to $\pi$.
\end{description}
\end{restatable}

Next, we move towards the general setting where the distance from a master list is arbitrary—specifically, cases where parameters like $k_a$ or $\text{depth}(S^a)$ are in $\omega(\log n)$ for some $a$.

\subsubsection*{Arbitrary distance from unanimity}

In the realm of arbitrary instances, where the distance from a master list is high, we present an array of algorithmic results based on a careful analysis of how early agents impact the potential value available to later agents.
Intuitively, in the restricted cases discussed above, the proximity to a master list provides a form of redundancy: there are multiple high-value \parts for an agent, ensuring that one agent's choice does not completely eliminate the high-value options for a successor.
However, when the distance is large, this redundancy vanishes. 
Scenarios such as those in \Cref{rem:need-for-rebundling,rem:impact-low-prop} can arise, where early choices destroy the viability of the instance for later agents.
Furthermore, rational agents are indifferent between distinct \parts of equal value. 
Consequently, an early agent might arbitrarily select a part that is the unique high-value part for a later agent, even when a non-critical alternative exists.
To navigate the gamut of possibilities arising from such scenarios, we identify sufficiency conditions (detailed in \Cref{sec:large-prop}) that allow us to guarantee every agent a high-value \part.

Our first set of results operates under the assumption that the algorithm knows the arrival order of agents, denoted by $\sigma=\seq{n}$, in advance, alongside the individual valuation functions $\{v_i\}_{i\in [n]}$.
Within this framework, \Cref{thm:all-pos-val-guarantee} addresses instances where agents value each item \emph{positively} (strictly greater than zero).
Next, \Cref{thm:large-prop-guarantee} generalizes this to instances where agents may value items at zero, provided that all agents have sufficiently large proportional shares.
This allows agents to obtain a comparable, if not better, guarantee.

Subsequently, in \Cref{thm:large-prop-guarantee-extended}, we demonstrate that the assumption of a fixed arrival order can be relaxed. 
We show that there exists a \emph{fair arrival order}, computable in polynomial time via a constructive greedy strategy, that yields a comparable guarantee.
This result is particularly relevant in fair division settings where valuations are known to the mechanism; it allows the mechanism and agents to agree upon a specific arrival order that provably achieves the desired guarantee.

We begin with the case of strictly positive valuations.

\begin{restatable}{theorem}{allposvalguarantee}
    \label{thm:all-pos-val-guarantee}
    Consider a fair division instance where the order of arrival of agents is known, and each agent values all items positively.
    There is an algorithm that guarantees each agent $i$ a part of value at least $\prop_i - \lceil \log n \rceil$.
\end{restatable}

The next result generalizes this to allow for zero-valued items, obtaining a stronger guarantee for each agent.
However, this improvement relies on all agents possessing sufficient ``baseline'' proportional shares, we refer to this assumption as {\it proportional shares being bounded}.
We note that the guarantee is meaningful only when $\prop_i$ is large enough to absorb the  logarithmic term.
If $\prop_i \le 2\lceil \log i \rceil + 1$, the guaranteed value becomes non-positive, and such an agent could, in principle, be given an empty \part!
Conversely, if the proportional share marginally exceeds this threshold (i.e., $2\lceil \log i \rceil + 2$), the guarantee becomes strictly positive.

We believe that the proof of \Cref{thm:large-prop-guarantee} is the most technically involved among our algorithmic results. 
Accordingly, we treat the proof of \Cref{thm:all-pos-val-guarantee} as a conceptual warm-up, introducing the core techniques in a simplified setting. 
The subsequent analysis for \Cref{thm:large-prop-guarantee} then highlights the refinements required to handle zero-valued items.

\begin{restatable}{theorem}{largepropguarantee}
    \label{thm:large-prop-guarantee}
    Consider a fair division instance where the order of arrival of agents is $\langle 1, \dots, n \rangle$, and each agent $i$ has a proportional share at least $2 (\lceil \log i \rceil + 1)$.
    There is an algorithm that guarantees each agent $i$ a part of value at least $\prop_i - 2 \lceil \log i \rceil - 1$. 
\end{restatable}

We note that the logarithmic bound used in \Cref{thm:large-prop-guarantee} serves as a convenient worst-case estimate. 
We can refine this analysis by introducing a specific parameter $g(i)$ (defined formally in \Cref{sec:large-prop}), where $1 \le g(i) \le \lceil \log i \rceil + 1$.
Substituting this into the analysis yields the following corollary, which provides tighter guarantees with weaker requirements.

\begin{restatable}{corollary}{corlargepropguarantee}
    \label{cor:large-prop-guarantee}
    Consider a fair division instance where the order of arrival of agents is $\langle 1, \dots, n \rangle$, and each agent $i$ has a proportional share at least $2 g(i)$.
    There is an algorithm that guarantees each agent $i$ a \part of value at least $\prop_i - 2 g(i) + 1$. 
\end{restatable}

We conclude by discussing the robustness of \Cref{alg:large-prop} (that proves \Cref{thm:large-prop-guarantee}) in \Cref{rem:robustness-alg-largeprop}.
Crucially, we note that the assumption of large proportional shares is primarily an analytical tool for certifying strong guarantees, rather than a mechanical prerequisite for the algorithm's execution.

\subsubsection*{Hypergraph setting.} 
The next result pertains to the \emph{hypergraph setting}, a generalization of the \emph{graphical setting} introduced by Christodoulou \etal~\cite{DBLP:conf/sigecom/0001FKS23}.
In this setting, agents correspond to vertices and items to hyperedges; an agent values an item positively only if they are incident to the corresponding hyperedge.
Note that the valuation functions can be asymmetric, as in agents that are part of an edge can value the corresponding item differently.
This setting has gained significant traction in fair division research~\cite{DBLP:journals/corr/abs-2506-20317,DBLP:conf/ijcai/0027WL024,DBLP:conf/ifaamas/SgouritsaS25,DBLP:conf/ijcai/DeligkasEGK25,hsu2025efxorientationsmultigraphs,DBLP:journals/corr/abs-2404-13527,10.5555/3709347.3743514,DBLP:conf/aldt/MisraS24,bhaskar2024efxallocationsmultigraphclasses}, as it represents a reasonable restriction where every item is relevant to at most a small number of agents.

In our analysis, the governing parameter is the \emph{influence set} of an agent. 
For an agent $a$, the influence set consists of all agents (including $a$) who value at least one item positively that $a$ also values positively.
Formally, this corresponds to a subset closed neighborhood of $a$ in the hypergraph.
Note that the size of the hyperedges and the closed neighborhoods can be unbounded.
Let $D$ denote the maximum cardinality of any influence set in the instance.

\begin{restatable}{theorem}{boundedinfluenceguarantee}
    \label{thm:bounded-influence-guarantee}
    Consider a fair division instance where the order of arrival of agents is known, the influence of each agent is bounded by $D$, and $v_a([m])/D \ge 2 \lceil \log (2D-1) \rceil + 2$ for each agent $a \in [n]$.
    Then, there is an algorithm that guarantees each agent $a$ a part of value at least:
    \begin{itemize}
        \item $v_a([m])/D - 2\lceil \log(2D-1) \rceil - 1$, if $n \le 2D-1$ or $a$ is among the last $2D-1$ agents;
        \item $v_a([m])/D - 1$, if $a$ is among the first $n-(2D-1)$ agents.
    \end{itemize}
\end{restatable}

This result demonstrates that in the hypergraph setting, we can obtain guarantees significantly superior to those in the general setting (\Cref{thm:large-prop-guarantee}).
Most strikingly, when $n \gg 2D$, the quantity $v_a([m])/D$ far exceeds the standard proportional share $\prop_a = v_a([m])/n$.
Consequently, the first $n-(2D-1)$ agents receive a value strictly greater than their proportional share.
Even for the remaining agents (or when $n \le 2D-1$), the guarantee remains superior because the error term scales with $\log(2D)$ rather than $\log n$.

\subsubsection*{Bounded indifference.} 
Finally, we examine instances where agents have bounded indifference. 
We define the \emph{tie size} $t$ as the maximum number of items with the same positive value for any agent.
In such scenarios, we can express the guarantee directly in terms of $t$, provided the agents satisfy a large proportional share assumption.

\begin{restatable}{theorem}{nolargetiesguarantee}
    \label{thm:no-large-ties-guarantee}
    Consider a fair division instance where the order of arrival of agents is known, each agent has at most $t$ non-zero items of the same value, and each agent has a proportional share of at least $\lceil \frac{t+3}{2} \rceil$.
    There is an algorithm that guarantees each agent $i$ a part of value at least $\prop_i - \lceil \frac{t+3}{2} \rceil$.
\end{restatable}

\subsubsection*{A ``fair'' arrival order} 
Our final technical result addresses the design of the arrival sequence itself. 
We show that we can compute a specific arrival order such that the execution of \Cref{alg:large-prop} results in every agent receiving a high-value \part.

\begin{restatable}{theorem}{largepropguaranteeextended}
    \label{thm:large-prop-guarantee-extended}
    For any fair division instance, there exists an arrival order of agents, denoted by $\tau$, such that each agent $\tau_i \in [n]$, when picking autonomously and rationally, gets a part of value at least $\prop_{\tau_i} - 2 \lceil \log i \rceil - 2$. 
    Moreover, $\tau$ can be computed in polynomial time.
\end{restatable}

\paragraph{Organization of the paper.} 
At its core, all our algorithms utilize the \rr method (\Cref{alg:round-robin}); its essential properties are established in \Cref{sec:prelim}.
The related work is discussed in \Cref{sec:related-works}.
The remainder of the paper is organized as follows:
\begin{itemize}[wide=0pt]
    \item \Cref{sec:all-pos} initiates the dynamic partitioning framework, focusing on strictly positive valuations.
    \item \Cref{sec:large-prop} extends the dynamic framework to arbitrary valuations with proportional shares \Ma{being bounded}.
    \item \Cref{sec:bounded-influence-guarantee} examines the hypergraph setting, deriving better guarantees when influence is bounded.
    \item \Cref{sec:no-large-ties} covers the bounded indifference setting, where value indifference is limited.
    \item \Cref{sec:ordered-additive-relaxations} analyzes the static partitioning framework for instances close to unanimity.
    \item Finally, \Cref{sec:one-sided-lower-bound} establishes the impossibility of improving universal guarantees.
\end{itemize}

We emphasize that our results are tight with respect to the proof strategies employed; improving the bounds is likely to require a significant departure from current techniques.
Moreover, most of our guarantees are consequences of stronger structural results. 
Our proofs establish that each agent gets to choose a part that is at most a certain number of items (not present in this part) away from proportionality. 
This bound combined with the fact that item values are in $[0,1]$ directly imply our value guarantees.
\section{Preliminaries}\label{sec:prelim}



We use $M=(\M{1}, \dots, \M{n})$ to denote a partition and $\M{i}$ to denote its $i$-th part. 
In the following description, we assume without loss of generality that the agents are indexed $1, \dots, n$ and form the ``picking sequence'' of the algorithm.

\ifthenelse{\equal{\algostyle}{algorithm2e}}{
    \begin{algorithm}
  \caption{Round-Robin}\label{alg:round-robin}
  \SetKwFunction{FRoundRobin}{\textsc{RoundRobin}}
  \SetKwProg{Fn}{Function}{:}{}
  
  \Fn{\FRoundRobin{$[m], v_1, \dots, v_n$}}{
    $U \leftarrow [m]$ \tcp*{the set of unallocated items}
    $M \leftarrow (\emptyset, \dots, \emptyset)$ \tcp*{the partition we will construct}
    $p \leftarrow 1$ \tcp*{agent who picks next}
    
    \While{$U \ne \emptyset$}{
      $g \gets \text{a most preferred item of agent } p \text{ in } U$\;
      $\M{p} \gets \M{p} \cup \{g\}$\;
      $U \gets U \setminus \{g\}$\;
      $p \leftarrow (p \mod n) + 1$\;
    }
    \KwRet $M$\;
  }
\end{algorithm}
}{}

\begin{proposition}[\cite{DBLP:conf/sigecom/ConitzerF017,DBLP:journals/teco/CaragiannisKMPS19}]
\label{prop:RR-gives-universal-prop} The output of \rr (\Cref{alg:round-robin}), denoted by $M=(\M{1}, \ldots, \M{n})$, has the following properties:

\begin{enumerate}
    \item For each agent $i$ who picks $i$-th in the sequence, $v_i(\M{i}) \ge v_i([m])/n - 1 = \prop_i - 1$.

    \item No agent envies any part corresponding to agents who picked later in the sequence.
    That is, for each $\{i,j\} \sse [n]$ such that $i < j$, we have $v_i(\M{i}) \ge v_i(\M{j})$.
\end{enumerate}
    
\end{proposition}

In our proofs, we sometimes consider the output of \rr as a matrix where the item in row $r$ and column $c$ is the $r$-th pick of agent $c$.

Next, we discuss an example that shows that frequent rebundling helps to protect value for later agents.
\begin{remark}[Power of rebundling]
\label{rem:need-for-rebundling}
    The following example illustrates the general necessity of rebundling (or actively moving items between parts) in a partition.
    Consider $n$ agents, with the same proportional share $\prop$, arriving in the order $1, \dots, n$, where $n$ is even and $\prop$ is an integer.
    Suppose that the initial partition is $\M{1}, \dots, \M{n}$ with the following valuations:
    \begin{itemize}[itemsep=0pt]
        \item For each agent $i \in [n-1]$, $v_i(\M{i}) = v_i(\M{i+1}) = \frac{n}{2} \cdot\prop$, and $0$ elsewhere. 
        \item The agent valuations are binary.
        \item For agent $n$, $v_n(\M{n}) = n \cdot \prop$ and $0$ elsewhere.
    \end{itemize}

    Consider the case of no rebundling.
    We assume that agents break ties among equal valued bundles by picking the higher index.
    \begin{itemize}[itemsep=0pt]
        \item Agent 1 picks $\M{2}$.
        \item Agent 2, finding $\M{2}$ gone, picks $\M{3}$.
        \item This chain continues until agent $n-1$ picks $\M{n}$.
        \item Agent $n$ arrives to find $\M{n}$ taken and is forced to pick $\M{1}$ (value 0).
    \end{itemize}
    
    Suppose that we are allowed to rebundle between agent arrivals.
    When agent $i \in [n-1]$ arrives, we transfer two item of value $1$ for agent $i$ from $\M{i+1}$ to $\M{i}$.
    \begin{itemize}[itemsep=0pt]
        \item Agent $i$'s valuation becomes $v_i(\M{i}) \ge \frac{n}{2}\cdot\prop$.
        \item Agent $i$ strictly prefers and picks $\M{i}$.
        \item This preserves $\M{i+1}$ for the subsequent agent.
        \item Finally, agent $n$ freely picks $\M{n}$, which is of value $n\cdot\prop-2$.
    \end{itemize}
    In this new execution, every agent $i \in [n-1]$ picks a bundle of value at least $\frac{n}{2} \prop_i$ and agent $n$ picks a bundle of value $n\cdot \prop-2$. We note that this scenario can be amplified to obtain a scenario, where an unbounded number of agents may receive low value without rebundling. \qed
\end{remark}

\section{Related work}
\label{sec:related-works}

For a comprehensive overview on proportional allocations, and fair division in general, we refer the reader to the surveys \cite{DBLP:journals/ai/AmanatidisABFLMVW23,DBLP:journals/corr/abs-2307-10985}.

Manurangsi and Suksumpong \cite{DBLP:journals/tcs/ManurangsiS22} study the notion of \textit{consensus 1/$n$ division}.
An allocation $M=(M[1], \dots, M[n])$ is said to be a \emph{consensus 1/$n$-division up to $c$ goods} if, for every agent $a$ and every pair of parts $M[i]$, $M[i']$, there exists $B \subseteq M[i']$ with $|B| \le c$ such that $a$ values $M[i]$ no less than $M[i'] \setminus B$.
They show that a consensus fair division always exists when $c = \bigoh(\sqrt{n})$.
Observe that in any such allocation, every part is of value at least $\prop_a - c = \prop_a - \bigoh(\sqrt{n})$ for each agent $a$.
Manurangsi and Meka \cite{doi:10.1137/1.9781611978964.20Manurangsi26} show that it is asymptotically tight.

Our model contrasts with the recent structural results in allocation without monetary transfers, due to Babaioff and Morag~\cite{DBLP:conf/sigecom/BabaioffM25}. 
They establish that any deterministic, strategy-proof, neutral, and non-bossy mechanism must be a \emph{serial quota mechanism}, where agents sequentially choose \emph{any} subset of items (up to a fixed size) from those remaining. 
This highlights a fundamental trade-off: insisting on strategy-proofness restricts the mechanism to simple cardinality constraints. 
Our work explores the alternative side of this trade-off: by relaxing incentive compatibility, we can impose stricter structural constraints---requiring agents to select from a pre-defined partition---which enables us to guarantee strong fairness properties. 

In addition to the papers discussed in the introduction, we note that there is a substantial body of work on fair division of indivisible items centered around proportionality. 
In particular, the impossibility of exact proportionality, even under additive valuations, has led to a plethora of relaxations. 
Notable examples include proportionality up to one item (PROP1), introduced and analyzed by Conitzer \etal~\cite{DBLP:conf/sigecom/ConitzerF017}. 
A closely related line of work studies the maximin share (MMS) guarantee, introduced by Budish \cite{budish2011combinatorial} and further developed algorithmically by Procaccia and Wang \cite{DBLP:journals/jacm/KurokawaPW18}.
This literature predominantly considers centralized allocation algorithms that directly assign items to agents.

Sequential allocation mechanisms have been studied extensively, starting with classical picking sequences introduced by Brams and Taylor~\cite{DBLP:books/daglib/0017730}.\ma{reall? check!} \hide{and later analyzed from a computational perspective by Bouveret and Lang.} 
More recent work considers online fair division, where allocations must be made irrevocably as agents or items arrive. These models typically allocate single items and focus on efficiency or envy-based criteria. 
In contrast, our setting allows the algorithm to design and present bundles of items to agents, enabling qualitatively different proportionality guarantees.

A separate strand of work studies fair allocation in the absence of monetary transfers, often under additional structural or incentive constraints. 
It is known that no deterministic mechanism without payments can guarantee strong fairness notions such as proportionality or MMS under general additive valuations (see Padala and Gujar~\cite{DBLP:conf/webi/PadalaG21}). 
Related impossibility and trade-off results are shown by Psomas and Verma \cite{DBLP:conf/nips/0001V22}. These results motivate studying fairness guarantees in non-truthful, algorithmic models, such as the sequential bundle-choice framework considered in this work.\ma{do we need this para again?}


Our focus is on designing efficient algorithms that guarantee each agent a value as close as possible to their proportional share. 
We show that, even in this constrained sequential setting, carefully constructed bundle menus can substantially improve fairness outcomes compared to naive greedy allocation. 
Our results demonstrate that proportionality remains a meaningful and achievable objective in sequential, choice-driven allocation environments.
\section{When all items are positively valued}
\label{sec:all-pos}

In this section, we prove \Cref{thm:all-pos-val-guarantee}. 
Accordingly, we assume throughout this section that all valuations are strictly positive, i.e., $v_i(g) > 0$ for all agents $i$ and items $g$.
Suppose that the agents arrive in increasing order of their labels: agent $1$ first, followed by agent $2$, and so on.
We present a recursive procedure that guarantees each agent $i$ a \part of value at least $\prop_i - \lceil \log n \rceil$.

Note that if the number of items $m$ is at most $n \lceil \log n \rceil$, then the proportional share of any agent is at most $\lceil \log n \rceil$; consequently, any partition satisfies the desired guarantee.
Thus, from now on, we assume that $m > n \lceil \log n \rceil$.

\paragraph{Algorithm overview}
The algorithm proceeds recursively.
The core idea is to process $\lceil n/2 \rceil$ agents in each recursive stage, ensuring they pick \parts of high value upon arrival.

\textbf{Preprocessing phase. }
We first compute a partition $M = (\M{1}, \dots, \M{n})$ using \rr, simulating picks in the same order as that of the agents' arrival.

\textbf{Picking phase. }
The agents arrive sequentially. 
As agent $i \in [\lfloor n/2 \rfloor]$ arrives, we transfer the top item\footnote{We note that transferring \emph{any} item suffices, but we select the top item for the sake of simplicity.} from the \emph{donor} \part $\M{i+\lfloor (n+1)/2 \rfloor}$ to $\M{i}$.
We then present the unpicked \parts to agent $i$. 
This transfer strictly increases the value of $\M{i}$, thereby making $\M{i}$ the unique favorite \part for agent $i$ among all unpicked \parts.
Thus, agent $i$ picks $\M{i}$.
Finally, we recurse on the remaining agents and items.

\newcommand{\don}{\texttt{don}}
\newcommand{\rec}{\texttt{rec}}
\newcommand{\offset}{\texttt{offset}}

\ifthenelse{\equal{\algostyle}{algorithm2e}}{
\begin{algorithm}
  \caption{Algorithm for positive valuations}\label{alg:all-pos-val}
  \SetKwFunction{FAllPosVal}{AllPosVal}
  \SetKwFunction{FMain}{Main}
  \SetKwProg{Fn}{Function}{:}{}
  
  \Fn{\FAllPosVal{$U, v_1, \dots, v_n, \texttt{stage}$}}{
      \If{$n=1$}{
          Present $U$ to the sole agent\;
          \textbf{Terminate}\;
      }
      \tcp{Phase 1: Preprocessing}
      $M \leftarrow \textsc{RoundRobin}(U, v_1, \dots, v_n)$\;
      $\texttt{mid} \gets \lfloor \frac{n}{2} \rfloor$\;
      $\offset \gets \lfloor \frac{n+1}{2} \rfloor$\;

      \tcp{Phase 2: Picking}
      \For{$i \leftarrow 1$ \KwTo $\texttt{mid}$}{
          $\rec \gets i$ \tcp*{Index of receiver agent}
          $\don \gets i + \offset$ \tcp*{Index of donor agent}
          $g_i \leftarrow \text{top item of agent } \don \text{ in } \M{\don}$\;
          $\M{\don} \gets \M{\don} \setminus \{ g_i \}$ \tcp*{Remove $g_i$ from donor}
          \label{line:remove-gi}
          $\M{\rec} \leftarrow \M{\rec} \cup \{g_i\}$ \tcp*{Add $g_i$ to receiver}
          \label{line:add-gi}
          Present the unpicked \parts in $M$ to agent $\rec$\;
      }
      
      \If{$n$ is odd}{
          Present the unpicked \parts in $M$ to agent $\lceil \frac{n}{2} \rceil$\;
      }
      
      $U_{\texttt{next}} \leftarrow \bigcup_{j: \M{j} \text{ is unpicked}} \M{j}$\;
      \FAllPosVal{$U_{\texttt{next}}, v_{\lceil \frac{n}{2} \rceil+1}, \dots, v_n, \texttt{stage}+1$} \tcp*{Recurse}
  }
  
  \Fn{\FMain{$[m], v_1, \dots, v_n$}}{
      \FAllPosVal{$[m], v_1, \dots, v_n, 1$}\;
  }
\end{algorithm}
}{}

\begin{figure}[h]
    \centering
    \includegraphics[width=0.9\linewidth]{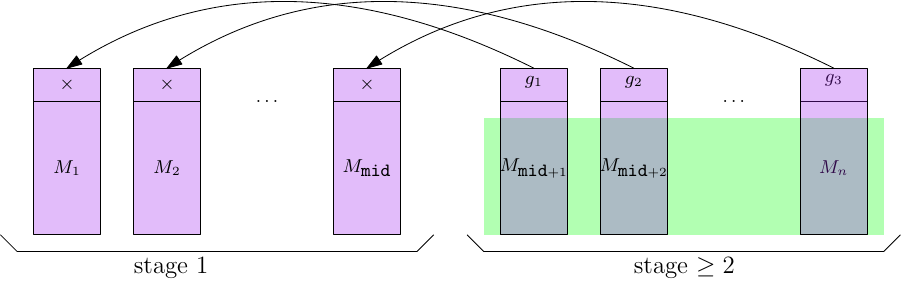}
    \caption{
    An illustration of the execution of \Cref{alg:all-pos-val}. 
    At the $i$-th step of the picking phase, the top item $g_i$ from the donor part $\M{\texttt{mid}+i}$ in the second half (stage $\ge 2$) is moved to the corresponding receiver part $\M{i}$ in the first half (stage 1).
    The shaded green region indicates the remaining items that form the sub-instance for the recursive step.
    }
    \label{fig:all-pos-val}
\end{figure}

The following lemma characterizes the behavior of the agents at any given stage $k$, establishing both the correctness of the \part selection and the resulting guarantee.

\begin{restatable}{lemma}{allposvalue}
    \label{lem:all-pos-val}
    Fix any recursive stage $k \le \lceil \log n \rceil$. 
    Any agent $i$ who picks in stage $k$ receives a \part of value at least $\prop_i - k$.
    Also, when the unpicked \parts are presented to agent $i$, they choose $\M{i}$.
\end{restatable}
\begin{proof}
    We first argue for the base case $k=1$.
    Via induction on $i$, we show that agent $i$ precisely picks $\M{i}$ when the unpicked \parts in $M$ are presented. 

    Recall that an item $g_i$ is added to $\M{i}$ in Line \ref{line:add-gi}.
    For the base case, consider $i=1$.
    By the guarantee of \rr, we have that $\M{1} \setminus \{g_1\}$ (the original \part) is at least as valuable as any other \part $\M{j}$ ($j > 1$) for agent $1$, and is of value at least $\prop_1 - 1$.
    The addition of the positively valued item $g_1$ makes $\M{1}$ strictly better than its original state.
    Thus, $\M{1}$ becomes the uniquely most valuable \part among the unpicked \parts, and agent $1$ picks it.
    
    Suppose this holds up to agent $i-1$.
    The first $i-1$ agents have picked the first $i-1$ \parts.
    By the guarantee of \rr, $\M{i} \setminus \{g_i\}$ is at least as valuable as any $\M{j}$ ($j > i$) for agent $i$, and is of value at least $\prop_i-1$.
    The addition of the positively valued item $g_i$ makes $\M{i}$ the uniquely most valuable unpicked \part.
    Thus, any agent $i \in \{1, \dots, \lfloor \frac{n}{2}\rfloor\}$ picks $\M{i}$.
    
    If $n$ is odd, agent $\lceil \frac{n}{2} \rceil$ also finds their \part uniquely best, and of value at least $\prop_{\lceil \frac{n}{2} \rceil} - 1$. This is due to the initial \rr guarantees and the fact that all trailing agents ($\lceil \frac{n}{2} \rceil + 1, \dots, n$) have lost an item during this stage (decreasing the value of their \parts for everyone, including agent $\lceil \frac{n}{2} \rceil$).

    After this stage concludes, the remaining unpicked items form the set $U_{\texttt{next}}$.
    Consider the output of \rr as a matrix where the item in row $r$ and column $c$ is the $r$-th pick of agent $c$.
    The items \emph{assigned} to agents in the current stage constitute:    
    \begin{enumerate}
        \item The entire columns $1$ through $\lceil n/2 \rceil$.
        \item The first row of columns $\lceil n/2 \rceil+1$ through $n$ (used for transfers).
    \end{enumerate}
    Consequently, the set $U_{\text{next}}$ consists exactly of the columns $\lceil n/2 \rceil +1 \dots n$ with their first row removed.
    
    Crucially, the restriction of a \rr allocation to a subset of agents and a subset of items (in the same order) is equivalent to running \rr afresh on them.
    Thus, stage $k$ is similar to stage 1, but operating on a valuation matrix where the first $(k-1)$ rows have been removed from the perspective of the current agents.
    Recall that $m \ge n \lceil \log n \rceil$.
    The remaining items form a matrix where all but the last row is filled (because there are at least $\lceil \log n \rceil$ rows since $m > n \lceil \log n \rceil$, and the last row is filled from left to right, leaving the empty slots contiguous on the right).
    We note that this remaining partition $(\M{\lceil \frac{n}{2} \rceil+1}, \dots, \M{n})$ is equivalent to the output of a \rr call on the unpicked items with the remaining agents in the same order.
    Thus, by recursively applying the argument, any agent picking in stage $k$ has lost $k-1$ items from their designated \part in previous stages (due to being a donor in Line~\ref{line:remove-gi}) and receives all remaining items.
    Since the initial value of $\M{i}$ is at least $\prop_i-1$, the final value is at least $\prop_i - 1 - (k-1) = \prop_i - k$.
\end{proof}

Since the number of active agents halves at each stage, the maximum recursion depth is $\lceil \log n \rceil$.
Thus, the value of the \part received by any agent $i$ is at least $\prop_i - \lceil \log n \rceil$, which completes the proof of \Cref{thm:all-pos-val-guarantee}.

\hide{We note that while \rr in the preprocessing phase of \Cref{alg:all-pos-val} can be executed without the knowledge of the arrival order of agents, such an execution may lead to later agents in the sequence obtaining much less than their proportional share.\ma{Given the foll remark, why do we need this ?remove ?}}

\begin{remark}[Importance of Knowing the Arrival Order]
    We note that execution of \rr in the preprocessing phase of \Cref{alg:all-pos-val} without the knowledge of the arrival order of agents may lead to later agents in the sequence obtaining much less than their proportional share.
    
    Consider an instance with three agents where the \rr \parts $\M{1}, \M{2}, \M{3}$ are computed using the picking order $\langle 2, 1, 3 \rangle$, but the agents actually arrive in the order $\langle 1, 2, 3 \rangle$.
    Suppose the valuations are as follows:
    \begin{itemize}
        \item Agent 1 values $\M{2}$ at $1.5 \prop_1 + 1$, $\M{1}$ at $1.5\prop_1 - 1 - \epsilon_1$, and $\M{3}$ at $\epsilon_1$ (where $\epsilon_1 > 0$ is small).
        \item Agent 2 values $\M{2}$ at $3\prop_2 - 2\epsilon_2$ and the other two \parts at $\epsilon_2$ (where $\epsilon_2 \approx 0$).
    \end{itemize}
    The algorithm aims to assign $\M{1}$ to agent 1. 
    To ensure agent 1 picks $\M{1}$ autonomously, it attempts to make $\M{1}$ the unique favorite by transferring an item $g$ from the donor \part $\M{3}$ to $\M{1}$.
    However, since $\M{3}$ contains items of total value $\epsilon_1$, the transfer increases the value of $\M{1}$ by at most $\epsilon_1$.
    Thus, the augmented $\M{1}$ has value at most $(1.5\prop_1 - 1 - \epsilon_1) + \epsilon_1 = 1.5\prop_1 - 1$.
    Agent 1 compares this against $\M{2}$ (value $1.5\prop_1 + 1$) and rationally picks $\M{2}$.
    Consequently, agent 2 is left with only $\epsilon_2$-valued parts, which violates their guarantee. \qed
\end{remark}

\section{When each agent's proportional share is bounded}
\label{sec:large-prop}

In this section, we prove \Cref{thm:large-prop-guarantee}.
We consider the general setting where agents have arbitrary additive valuations (potentially including zeros), provided that every agent's proportional share is sufficiently large.

Recall that the agents arrive in the order $\langle 1, \dots, n \rangle$.
We assume that the proportional share of every agent $i$ satisfies the following lower bound: $\prop_i \ge 2 (\lceil \log i \rceil + 1)$.

This condition is intrinsic to our algorithm, which relies on agents possessing sufficient value to facilitate the necessary partition updates; it is not designed to handle agents with negligible shares.
However, we note that for any agent falling below this threshold, the target guarantee $\prop_i - 2 (\lceil \log i \rceil + 1) + 1$ would be at most 1.

\paragraph{Algorithm Overview}
\label{sec:largeprop-overview}
The core challenge in this general setting, unlike the strictly positive case (\Cref{alg:all-pos-val}), is that an agent might value the items in the ``donor'' \parts at zero, making the simple item transfer strategy ineffective.
To handle this, we introduce a \emph{rebundling} technique: if we cannot transfer a positively valued item from a donor to the receiver (agent $i$) to make their \part $\M{i}$ the unique favorite, we instead swap items between the unpicked \parts other than $\M{i}$ to ensure that $\M{i}$ emerges as the unique favorite; see \Cref{fig:large-prop} for illustration.

The algorithm proceeds recursively.
We define a \textit{stage} as a single level of recursion.
The following is a high-level overview of a single recursive stage:
\begin{enumerate}
    \item \textbf{Preprocessing phase:} 
    The first $\lfloor n/2 \rfloor$ agents are processed in this stage.
    We compute an initial partition $M = (\M{1}, \ldots, \M{n})$ using \rr with the picking order $\langle 1, \dots, n \rangle$.
    
    \item \textbf{Picking phase:}
    Agents $1$ through $\lfloor n/2 \rfloor$ arrive sequentially.
    When agent $i$ arrives, our goal is to ensure that $\M{i}$ is their uniquely preferred \part.
    We examine the items in the ``donor'' \parts (specifically, the top two picks of agents in the range $\lfloor n/2 \rfloor + 1, \dots, n$).
    \begin{enumerate}
        \item \textbf{Transfer:} 
        If there exists a donor item $g_i$ that agent $i$ values strictly positively, we transfer it to $\M{i}$.
        This strictly increases $v_i(\M{i})$ while leaving the value of other \parts unchanged or lower, making $\M{i}$ the unique favorite.
        
        \item \textbf{Rebundling:} 
        If there is no such item, it implies that agent $i$ values all available donor items as zero.
        We identify any \parts that currently provide the same value as $\M{i}$.
        We create a temporary partition $M'$ by swapping the top item of each such same-valued \part with a zero-valued item from a donor \part.
        This operation strictly decreases the value of these same-valued \parts while $\M{i}$ remains unchanged.
        We present the \parts in $M'$ to agent $i$, for whom the unique favorite is guaranteed to be $\M{i}$.
    \end{enumerate}
    
    \item \textbf{Recursion:} 
    After all agents in the current set have picked their \parts, the remaining unpicked items are pooled, and the algorithm recurses on the remaining agents with the stage counter incremented.
\end{enumerate}

\ifthenelse{\equal{\algostyle}{algorithm2e}}{
\begin{algorithm}
  \caption{Algorithm for bounded proportional shares}\label{alg:large-prop}
  \SetKwFunction{FLargeProp}{BoundedProp}
  \SetKwFunction{FMain}{Main}
  \SetKwProg{Fn}{Function}{:}{}
  
  \Fn{\FLargeProp{$U, v_1, \dots, v_n, \texttt{stage}$}}{
      \If{$n=1$}{
          Present $U$ to the sole agent\;
          \textbf{Terminate}\;
      }
      \tcp{Phase 1: Preprocessing}
      $M \leftarrow \textsc{RoundRobin}(U, v_1, \dots, v_n)$\;
      $\texttt{mid} \gets \lfloor \frac{n}{2} \rfloor$\;

      \tcp{Phase 2: Picking}
      \For{$i \leftarrow 1$ \KwTo $\texttt{mid}$}{
          \tcp{Agent $i$ arrives to pick}
          
          \If{$i = \texttt{mid}$}{
              \label{line10} 
              Add to $\M{i}$ all items $\M{j,a}$ where $j \in \{1,2\}$, $a \in \{\texttt{mid}+1, \ldots, n\}$ and $\M{j,a} \ne \perp$\;
              \tcp{Collect all unpicked items from top 2 rows of later agents}
              \label{line11} Present the unpicked \parts in $M$ to agent $i$\;
              \textbf{Continue}\;
          }

          \tcp{Check for a positively valued item in top 2 rows of donors}
          \If{$\exists~a \in \{\texttt{mid} + 1, \dots, n\}, j \in \{1,2\}$ s.t. $\M{j,a} \neq \perp$ and $v_i(\M{j,a}) > 0$}{
              \label{choice-of-item}
              Let $g_i \gets \M{j,a}$\;
              $\M{j,a} \gets \perp$ \tcp*{Remove item from donor and mark as taken}
              $\M{i} \leftarrow \M{i} \cup \{g_i\}$\;
              Present the unpicked \parts in $M$ to agent $i$\;
          }
          \Else{
              Construct temporary partition $M'$ via Rebundling:\;
              \quad $M'[q] = \M{q}$ for each $q \in [n]$\;
              \quad Swap top items: $M'[1, i+k] \leftrightarrow \M{1, \texttt{mid}+k}$ for $k \in \{1, \dots, \texttt{mid} - i\}$\;
              Present the unpicked \parts in $M'$ to agent $i$\;
          }
      }
      \tcp{Phase 3: Recursion}
      $U_{\texttt{next}} \leftarrow \bigcup_{j: \M{j} \text{ is unpicked}} \M{j}$\;
      \FLargeProp{$U_{\texttt{next}}, v_{\texttt{mid}+1}, \dots, v_n, \texttt{stage}+1$}\;

  }
  
  \Fn{\FMain{$[m], v_1, \dots, v_n$}}{
      \FLargeProp{$[m], v_1, \dots, v_n, 1$}\;
  }
\end{algorithm}
}{}

\begin{figure}
    \centering
    \includegraphics[width=0.9\linewidth]{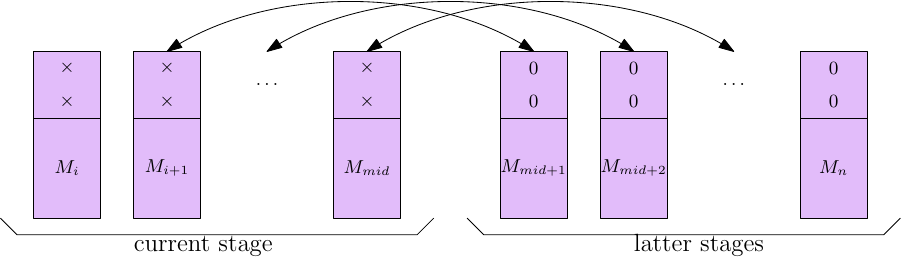}
    \caption{Illustration of the rebundling technique when no positive transfer is possible.}
    \label{fig:large-prop}
\end{figure}

\paragraph{Analysis of the guarantee}To analyze the performance of the algorithm, we first establish a lower bound on the value an agent receives when participating in a specific recursive stage $k$.

\begin{restatable}{lemma}{largepropguaranteestage}
\label{lem:large-prop-gurantee-stage}
Suppose that for every agent $i$, $\prop_i \ge 2 \lceil \log i \rceil + 2$.
Fix any recursive stage $k \le \lceil \log n \rceil + 1$. 
Any agent $i$ who picks in stage $k$ receives a \part of value at least $\prop_i - 2k+1$.
Moreover, when the unpicked \parts are presented to agent $i$, they choose the \part $\M{i}$.
\end{restatable}

\begin{proof}
    We first analyze the case $k=1$.
    Analogous to the proof of \Cref{lem:all-pos-val}, we proceed via induction on $i$ to show that agent $i$ precisely picks $\M{i}$ when the unpicked \parts are presented.
    
    For the base case, consider $i=1$.
    We break the analysis into two cases depending on whether there is a positively valued item $g_1$ as stated in Line~\ref{choice-of-item} or not.
    \begin{itemize}
        \item 
        Suppose there exists an item $g_1$ in the top 2 rows of some donor \part $\M{a}$ (where $a > \texttt{mid}$) such that $v_1(g_1) > 0$.
        By the properties of \rr, the \part $\M{1}$ is already at least as valuable to agent $1$ as any other \part $\M{j}$ ($j > 1$), and satisfies $v_1(\M{1}) \ge \prop_1 - 1$.
        The transfer of the positively valued item $g_1$ to $\M{1}$ strictly increases its value, making $\M{1}$ the unique favorite. 
        Thus, agent 1 picks $\M{1}$.
        \item 
        Suppose no such item exists.
        This implies that agent $1$ values all items among the top two picks of the donor \parts $\M{\lceil \frac{n}{2} \rceil +1}, \dots, \M{n}$ at $0$.
        The algorithm constructs a temporary partition $M'$ by swapping items.
        For each $j \in \{1, \ldots, \texttt{mid}-1\}$, we swap the top item of $\M{1+j}$ with that of $\M{\texttt{mid} + j}$.
        This ensures for all $j \in [1, \texttt{mid}-1]$, both $M'[1+j]$ and $M'[\texttt{mid}+ j]$ have value strictly less than that of $\M{1}$ because:
        \begin{itemize}
            \item the top two picks of $\M{\texttt{mid}+j}$ are of value zero (otherwise the algorithm would have picked one of them as $g_1$);
            \item the top two picks of $\M{1}$ are of non-zero value (since $\prop_1 \ge 2$ and $\M{1}$ is computed via \rr).
        \end{itemize}
        Note that $M'[n]$ is also strictly less valuable than $M'[1]$.
        Thus, agent $1$ picks $M'[1]$ (i.e., $\M{1}$), which has value at least $\prop_1 - 1$.
    \end{itemize}

    Consider any $1 < i < \texttt{mid}$.
    Suppose the claim holds for all agents prior to $i$.
    Thus, the first $i-1$ \parts in $M$ have been picked, and the unpicked \parts are $(\M{i}, \dots, \M{n})$.
    When agent $i$ arrives, the logic is identical to the base case. 
    We analyze the same two cases depending on whether there exists an item $g_i$ as stated in Line~\ref{choice-of-item} or not.
    If there exists such a $g_i$, adding it to $\M{i}$ makes $\M{i}$ the unique favorite.
    If not, the rebundling step strictly lowers the value of all same-valued \parts by replacing their top items with zero-valued items from the donors, while leaving $\M{i}$ unchanged.
    Therefore, agent $i$ picks $\M{i}$.

    Finally, consider the last agent of this stage, $i = \texttt{mid}$.
    By the property of \rr, the \part $\M{\texttt{mid}}$ is at least as valuable as any other unpicked \part.
    The algorithm adds to $\M{\texttt{mid}}$ all remaining items from the top two rows of the donor agents ($\texttt{mid}+1$ to $n$).
    Thus, $\M{\texttt{mid}}$ becomes the uniquely best \part, and agent $\texttt{mid}$ picks it.

    After stage $1$ concludes, the remaining unpicked items form the set $U_{\texttt{next}}$.
    Consider the output of \rr as a matrix where the item in row $r$ and column $c$ is the $r$-th pick of agent $c$.
    The items \emph{assigned} to agents in the current stage constitute:    
    \begin{enumerate}
        \item The entire columns $1$ through $\texttt{mid}$.
        \item The first 2 rows of columns $\texttt{mid}+1$ through $n$ (used for transfers or swaps).
    \end{enumerate}
    Consequently, the set $U_{\text{next}}$ consists exactly of the columns $\texttt{mid}+1 \dots n$ with their first 2 rows removed.
    
    Crucially, the restriction of a \rr allocation to a subset of agents and a subset of items (in the same order) is equivalent to running \rr afresh on them.
    Thus, stage $k$ is similar to stage 1, but operating on a valuation matrix where the first $2(k-1)$ rows have been removed from the perspective of the current agents.
    
    The value of $\M{i}$ after the \rr call in stage $1$ is at least $\prop_i - 1$.
    In stage $k$, the \part $\M{i}$ is missing its top $2(k-1)$ items compared to the original allocation.
    Since the value drops by at most the number of missing items, the final value is at least $\prop_i - 1 - 2(k-1) = \prop_i - 2k + 1$.
\end{proof}

Having established the guarantee for an agent conditioned on the stage in which they participate, the next step is to bound the specific stage that an agent arriving at position $i$ participates in.

\begin{restatable}{claim}{agentstagelargeprop}
    \label{clm:agent-stage-largeprop}
    In \Cref{alg:large-prop}, an agent arriving at position $i$ participates in a recursive stage $k$ satisfying $k \leq \ceil{\log i} + 1$. 
\end{restatable}

\begin{proof}
    We prove this by induction on $i$.
    If $i \le \lfloor n/2 \rfloor$, the agent participates in stage $1$.
    Since $i \ge 1$, we have $1 \le \lceil \log i \rceil + 1$.
    Thus, the statement holds.

    Suppose that $i > \lfloor n/2 \rfloor$. 
    In the algorithm, agent $i$ is passed on to the next recursive call and is assigned the index $i' = i - \lfloor n/2 \rfloor$.
    The stage in which agent $i$ participates is exactly $1$ plus the stage of the agent in the recursive call.
    By the induction hypothesis, we have that in the recursive call, agent $i$ participates in stage $k$ where 
    \[
    k = 1 + k' \le 1 + (\lceil \log (i-\lfloor n/2 \rfloor) \rceil + 1) = \lceil \log (i-\lfloor n/2 \rfloor) \rceil + 2.
    \]
    To complete the proof, we must show that $\lceil \log (i-\lfloor n/2 \rfloor) \rceil + 2 \le \lceil \log i \rceil + 1$.
    We first establish an inequality that holds for any integer $i \ge 2$.
    Let $q = \lceil \log i \rceil$.
    Then $2^{q-1} < i \le 2^q$, and thus $2^{q-2} < i/2 \le 2^{q-1}$.
    This implies that $\lceil \log \lceil i/2 \rceil \rceil = q-1 = \lceil \log i \rceil -1$.

    Since $i \le n$, we have $\lfloor i/2 \rfloor \le \lfloor n/2 \rfloor$. 
    Substituting this back:
    \begin{align*}
        \lceil \log (i-\lfloor n/2 \rfloor) \rceil + 2 &\le \lceil \log (i- \lfloor i/2 \rfloor) \rceil + 2 \\
        &= \lceil \log \lceil i / 2 \rceil \rceil + 2\\
        &= (\lceil \log i \rceil - 1) + 2 \\
        &= \lceil \log i \rceil + 1.
    \end{align*}
\end{proof}

We are ready to prove the main guarantee.
Consider any agent $i$. 
By \Cref{clm:agent-stage-largeprop}, this agent participates in some recursive stage $k$ where $k \le \lceil \log i \rceil + 1$.
By \Cref{lem:large-prop-gurantee-stage}, agent $i$ picks the \part $\M{i}$ with value at least $\prop_i - 2(\lceil \log i \rceil + 1) + 1 = \prop_i - 2\lceil \log i \rceil - 1$.
Using the premise that $\prop_i \ge 2(\lceil \log i \rceil + 1)$, we observe that this guaranteed value is strictly positive:
\[ 
    v_i(\M{i}) \ge 2(\lceil \log i \rceil + 1) - 2(\lceil \log i \rceil) - 1 = 1. 
\]
    
Moreover, the \part $\M{i}$ has lost exactly $2(k-1)$ items in the previous stages. 
With the lower bound on $\prop_i$ and upper bound on $k$, the ``effective'' proportional share of agent $i$ in the sub-instance in stage $k$ is at least:
\[ 
    \prop_i - 2(k-1) \ge 2(\lceil \log i \rceil + 1) - 2\lceil \log i \rceil = 2. 
\]
This guarantees that the sub-matrix in stage $k$ contains at least two rows (with the $i$-th column containing at least two non-zero items for agent $i$). 
This ensures that the "top 2 rows" required for the transfer and rebundling steps are always available.
Overall, we have proved \Cref{thm:large-prop-guarantee}.

We conclude the analysis of the guarantee by discussing the consequences of not meeting the proportional share condition of \Cref{thm:large-prop-guarantee}. 

\begin{remark}[Impact of Low Proportional Shares]
    \label{rem:impact-low-prop}
    It is important to note that a violation of the proportionality condition in \Cref{thm:large-prop-guarantee} by a single agent jeopardizes the guarantees for all subsequent agents.
    Consider the first agent $i$ who fails the condition (i.e., $\prop_i < 2(\lceil \log i \rceil + 1)$).
    The low valuation of agent $i$ implies the algorithm may fail to make $\M{i}$ the unique favorite (for instance, if they value all unpicked parts at zero).
    Consequently, agent $i$ might select a part intended for a later agent $j > i$.
    Since the guarantees for subsequent agents rely on the specific availability of their designated \parts (the inductive invariant), this deviation disrupts the recursive structure.
    Therefore, for an arbitrary ordering, the algorithm guarantees success only up to the first agent who violates the proportionality requirement. 
    We will return to the matter of low proportional share later, in \Cref{rem:robustness-alg-largeprop}, where we establish a weaker sufficient condition that allows the algorithm to yield strong guarantees.
    \qed
\end{remark}

Next, we discuss two ways in which our analysis for \Cref{thm:large-prop-guarantee} can be refined to obtain a stronger result: first by showing that a stronger guarantee exists, and secondly by reducing the proportional share required by each agent.

\subsection{Strengthening the result: \Cref{cor:large-prop-guarantee}}

The logarithmic bound used in \Cref{thm:large-prop-guarantee} serves as a convenient worst-case estimate. 
However, we can refine this analysis by using the exact stage numbers to provide a higher guarantee.

Suppose that $\sigma=\langle 1, \dots, n \rangle$ denotes the arrival order of the agents.
Let $g: [n] \rightarrow [\lceil \log n \rceil + 1]$ be the function that maps each agent to the recursive stage they participate in.
The function $g$ is defined recursively as:
\begin{align*}
    g(i) = 
    \begin{cases}
        1, &\text{if } i \le \lfloor n/2 \rfloor; \\
        1 + g(i - \lfloor n/2 \rfloor), &\text{if } i > \lfloor n/2 \rfloor.
    \end{cases}
\end{align*}
By \Cref{clm:agent-stage-largeprop}, we know that $g(i) \le \lceil \log i \rceil + 1$. 
We can refine the analysis of \Cref{thm:large-prop-guarantee} by using this function $g$.

Recall from the proof that an agent $i$ participating in stage $k = g(i)$ has lost exactly $2(k-1)$ items to agents in prior stages.
For the recursion to proceed successfully, we showed that it was sufficient that the agent's ``effective'' proportional share in the current stage was at least $2$.
This requires $\prop_i - 2(g(i)-1) \ge 2 \iff \prop_i \ge 2 g(i)$.
Previously, we satisfied this by imposing the coarser condition of $\prop_i \ge 2(\lceil \log i \rceil + 1)$.
However, the condition $\prop_i \ge 2 g(i)$ is sufficient.
This leads us to \Cref{cor:large-prop-guarantee}. 

The following highlights the conditions required for the algorithm's successful execution.

\begin{remark}[Robustness of the Algorithm]
    \label{rem:robustness-alg-largeprop}
    It is important to observe that \Cref{alg:large-prop} is robust even for agents who fall below the proportional share threshold (i.e., $\prop_i < 2 g(i)$), provided they assign non-zero value to a sufficient number of items.
    Intrinsically, the algorithm's execution relies only on the \emph{existence} of non-zero values to establish strict inequalities, not their magnitude.
    
    Consider the operations performed when agent $i$ arrives.
    Either a positively valued item is found in Line~\ref{choice-of-item}, allowing the algorithm to execute a simple item transfer to make their \part uniquely best; 
    or, if no such item exists, the algorithm performs the rebundling step.
    Note that the rebundling invariant requires only that the agent's current top two items have \emph{some} positive value.
    Swapping the top items between the \parts $M'[\texttt{mid}+k]$ and $M'[i+k]$ still strictly lowers the value of $M'[i+k]$ relative to the agent's \part, regardless of how small the original value of $M'[i]$ was.

    \begin{quote}
        Structurally, the algorithm executes successfully as long as each agent $i$ possesses at least $2g(i)$ items with non-zero value in $\M{i}$ (which is sufficient to execute the operations and cover the items removed in previous recursive stages).
    \end{quote}
    
    The assumption that $\prop_i \ge 2g(i)$ used in \Cref{cor:large-prop-guarantee} is simply a convenient sufficient condition to ensure this support exists.
    Thus, the assumption that agents have reasonable proportional shares is primarily a tool for achieving strong guarantees, rather than a strict prerequisite for the algorithm's execution. \qed
\end{remark}

With the results strengthened for instances where agents have large proportional shares, a natural question remains: \textit{Can we extend similar guarantees to \emph{general} instances, those that contain agents with low proportional shares?}

We address this in the following subsection: by enforcing a specific arrival order, we can in fact generalize this bound.

\subsection{Guarantees for general instances}

We now prove \Cref{thm:large-prop-guarantee-extended}. That is, for any instance, we can compute an arrival order that ensures the bound applies to all agents regardless of their proportional share.

\begin{proof}[Proof of \Cref{thm:large-prop-guarantee-extended}]
We prove this by defining a greedy procedure to compute $\tau$, and using \Cref{alg:large-prop} as the allocation algorithm.
Let $U$ be the set of currently unassigned agents (initially $U = [n]$).
We determine the agent $\tau_i$ for each position $i$ sequentially from $1$ to $n$.
We define the \textit{candidate} set $C_i \subseteq U$ for position $i$ consisting of agents whose proportional shares are sufficiently large: $C_i = \{a \in U: \prop_a \ge 2 (\lceil \log i \rceil + 1) \}$.

We proceed with the following selection strategy for each step $i$ from $1$ to $n$:
\begin{itemize}[wide=0pt]
    \item \textbf{Case 1: Candidates exist.} 
    If $C_i \ne \emptyset$, we set $\tau_i$ as an arbitrary agent $a \in C_i$, and remove $a$ from $U$.
    \item \textbf{Case 2: No candidates.} 
    If $C_i=\emptyset$, we fill $\tau_i$ and all subsequent positions with the remaining agents in $U$ in an arbitrary order.
    (Note that since the proportionality threshold $2(\lceil \log i \rceil + 1)$ is non-decreasing over positions, if $C_i$ is empty, no unassigned agent could satisfy the condition for any position $j > i$).
    Consider any agent $\tau_j$ placed in this phase.
    By definition, $\prop_{\tau_j} < 2 (\lceil \log j \rceil + 1)$.
    Thus, this evaluates to $\prop_{\tau_j} - 2 (\lceil \log j \rceil + 1) \le 0$, which is trivially satisfied by any \part.
\end{itemize}

We analyze the outcome of \Cref{alg:large-prop} for an agent $\tau_i$ based on how they were assigned:
\begin{itemize}
    \item \textbf{Agents assigned in Case 1:} 
    By construction, the order $\tau$ consists of a prefix of agents selected via Case 1, followed by a suffix of agents filled via Case 2.
    For any agent $\tau_i$ in the prefix, we have $\prop_{\tau_i} \ge 2 (\lceil \log i \rceil + 1)$.
    Since no agent prior to $i$ violates the condition, by \Cref{rem:impact-low-prop} we infer that \Cref{alg:large-prop} guarantees agent $\tau_i$ a part of value at least $\prop_{\tau_i} - 2 \lceil \log i \rceil + 1$.
    
    \item \textbf{Agents assigned in Case 2:}
    By construction, we have $\prop_{\tau_i} < 2 (\lceil \log i \rceil + 1)$.
    As noted in \Cref{rem:impact-low-prop}, there is no guarantee on the value of a part picked by $\tau_i$ (it may even be zero).
    However, since the target bound is non-positive, the algorithm trivially guarantees agent $\tau_i$ a part of value at least the desired guarantee.
\end{itemize}
Thus, we have proved the theorem.
\end{proof}

Finally, we combine the greedy ordering strategy established in \Cref{thm:large-prop-guarantee-extended} with the refined analysis of the recursion depth from \Cref{cor:large-prop-guarantee} to obtain our tightest result.

\begin{corollary}
\label{cor:large-prop-guarantee-extended}
For any fair division instance, there exists an arrival order of agents, denoted by $\tau$, such that each agent $\tau_i \in [n]$, when picking autonomously and rationally, gets a \part of value at least $\prop_{\tau_i} - 2 g(i) - 2$. 
Moreover, $\tau$ can be computed in polynomial time.
\end{corollary}






\section{When all agents have bounded influence}\label{sec:bounded-influence-guarantee}

In this section we prove \Cref{thm:bounded-influence-guarantee}.
Recall that the influence set of an agent $a$, denoted by $I(a)$, is defined as the set of all agents (including $a$) who value at least one item positively that $a$ also values positively. 
Formally, $I(a) = \{ j \in [n] : \exists g \in [m] \text{ s.t. } v_j(g) > 0 \land v_a(g) > 0 \}$.

Let $D = \max_{a} |I(a)|$ denote the maximum \emph{degree of influence} in the instance.
Clearly, $1 \le D \le n$. 
When $D \ll n$, we can achieve significantly better guarantees (\Cref{thm:bounded-influence-guarantee}).
Towards that, we assume that $v_a([m])/D \ge 2 \lceil \log (2D-1) \rceil + 2$, for each agent $a\in [n]$.
In this section, we consider the case where the influence of every agent is bounded by $D$.

We consider an auxiliary \emph{Modified \rr} procedure which we will use later for our algorithm.
The basic difference from the standard \rr is that here agents are removed from the picking order if they have no items of positive value remaining.

\paragraph{Algorithm Overview.}
We maintain a list of active agents $L$, which is initialized to $(1, \dots, n)$.
A pointer keeps track of the current picking agent in this list.
In each step, the algorithm considers the most preferred item $g$ of the current picking agent $p$ among the unallocated items in $U$.
If agent $p$ values item $g$ strictly positively, then the item is assigned to the part $\M{p}$ and removed from $U$. 
The pointer then advances to the next agent in the list.
If agent $p$ values their favorite remaining item at zero (implying they value \emph{all} remaining items at zero), they are permanently removed from $L$.
Finally, if $L$ is empty and any item remains unallocated, it is assigned arbitrarily.

\ifthenelse{\equal{\algostyle}{algorithm2e}}{
\begin{algorithm}
  \caption{Modified Round-Robin}\label{alg:mod-round-robin}
  \SetKwFunction{FModRoundRobin}{ModifiedRoundRobin}
  \SetKwProg{Fn}{Function}{:}{}
  
  \Fn{\FModRoundRobin{$v_1, \dots, v_n$}}{
      $U \leftarrow [m]$ \tcp*{the set of unallocated items}
      $M \leftarrow (\emptyset, \dots, \emptyset)$ \tcp*{the partition we will construct}
      $L \leftarrow (1, \dots, n)$ \tcp*{List of active agents}
      $\tt{ptr} \leftarrow 1$ \tcp*{Index in list L}
      
      \While{$U \ne \emptyset$ \textbf{and} $L \ne \emptyset$}{
          $p \gets L[\tt{ptr}]$\;
          $g \gets \text{most preferred item of agent } p \text{ in } U$\;
          
          \If{$v_p(g)>0$}{
              $\M{p} \gets \M{p} \cup \{g\}$\;
              $U \gets U \setminus \{g\}$\;
              $\tt{ptr} \leftarrow (\tt{ptr} \mod |L|) + 1$ \tcp*{Move to next agent}
          }
          \Else{
              Remove agent $p$ from $L$ \tcp*{Agent $p$ drops out}
              \If{$L \ne \emptyset$}{
                  \If{$\tt{ptr} > |L|$}{
                      $\tt{ptr} \leftarrow 1$ \tcp*{wrap around if we had removed the last agent}
                  }
              }
          }
      }
      $\M{1} \leftarrow \M{1} \cup U$ \tcp*{Dump leftovers that are of value 0 to everyone}
      \KwRet $M$\;
  }
\end{algorithm}
}{}

\begin{restatable}{lemma}{modrrguarantee}
    \label{lem:mod-rr-guarantee}
    The output of Modified \rr, denoted by $M=(\M{1}, \ldots, \M{n})$, has the following properties:
    \begin{enumerate}
        \item For each agent $a$ who picks $a$-th in the sequence, $v_a(\M{a}) \ge v_a([m])/D - 1$.
        \item No agent envies any part corresponding to agents who picked later in the sequence.
        That is, for each $\{a,b\} \sse [n]$ such that $a < b$, we have $v_a(\M{a}) \ge v_a(\M{b})$.
    \end{enumerate}
\end{restatable}



\begin{proof}
    Fix an agent $a$. Observe that in every round of the algorithm where agent $a$ participates, they pick an item before any agent $j > a$.
    Thus, in every round, agent $a$ picks an item at least as valuable as the item picked by $j$ in that same round (according to $v_a$).
    If agent $a$ drops out of the list, it is because they value all remaining items at 0. 
    Consequently, any items added to $\M{j}$ after $a$ drops out contribute 0 to $v_a(\M{j})$.
    Summing over all rounds, we have $v_a(\M{a}) \ge v_a(\M{j})$.
    
    Let $S_a = \{ g \in [m] : v_a(g) > 0 \}$ be the set of items $a$ values positively.
    Consider the algorithm's execution restricted to items in $S_a$.
    An item $g \in S_a$ is removed from $U$ only if some agent $j$ picks it.
    The algorithm ensures that if $j$ picks $g$, then $v_j(g) > 0$.
    Combined with $v_a(g) > 0$, this implies that $j$ belongs to the influence set $I(a)$.
    
    Since $|I(a)| \le D$, in the worst case, agent $a$ effectively competes against at most $D-1$ other agents for the items in $S_a$.
    Specifically, agent $a$ is guaranteed to pick at least one item for every $D$ items removed from $S_a$, except possibly in the very last round of picks involving items in $S_a$.
    In that final round, agent $a$ might ``miss'' their turn if the items run out, losing value equivalent to at most one item (specifically, at most $\max_{g} v_a(g)$).
    Summing over the rounds, we have:
    \[
    v_a(\M{a}) \ge \frac{1}{D} \sum_{g \in S_a} v_a(g) - \max_g v_a(g) \ge \frac{v_a([m])}{D} - 1.
    \]
\end{proof}

Next, we prove \Cref{thm:bounded-influence-guarantee} by combining the \emph{bounded influence} property with the \emph{rebundling} technique introduced in the previous section, as described in \Cref{alg:bounded-influence}. 
The core observation is that if $j \notin I(i)$, then agent $i$ values the entire part $\M{j}$ as zero.
If the number of agents is large relative to $D$, specifically $n \ge 2D-1$, we can guarantee that there are enough ``zero-value parts'' available to facilitate swaps that devalue any other equally-valued parts that agent $i$ might prefer.
The algorithm consists of three phases:
\begin{itemize}[wide=0pt]
    \item \textbf{Phase 1: Initialization via Modified \rr.}
    We first compute an initial partition $M$ using the \textsc{ModifiedRoundRobin} procedure. 
    This step ensures a good preliminary allocation of items with positive marginal utility while simultaneously identifying agents who receive a part of value zero.
    
    \item \textbf{Phase 2: Iterative Rebundling (first $n-(2D-1)$ agents).}
    The algorithm processes the initial set of agents ($i = 1$ to $n - (2D - 1)$) sequentially. 
    Consider the arrival of agent $i$. The goal is to ensure that their assigned part \M{i} is strictly preferred over all future parts \M{j} (for $j > i$).
    The algorithm classifies the future parts into two sets based on agent $i$'s valuation:
    \begin{itemize}
        \item \texttt{Equals}: Parts that agent $i$ values exactly equal to \M{i} (competitor parts).
        \item \texttt{Zeros}: Parts that agent $i$ values at exactly 0 (donor parts).
    \end{itemize}
    Relying on the property that there are few equal parts ($|\texttt{Equals}| < D$) and sufficiently many zero parts ($|\texttt{Zeros}| \ge D-1$), we perform \emph{swaps} to break ties. 
    Specifically, we construct a modified partition $M'$ for agent $i$ such that the $i$-th part remains \M{i}, but it becomes the unique maximum.
    While any equal part $\M{r} \in \texttt{Equals}$ exists, we identify agent $i$'s most preferred item in \M{r} and swap it with an item from a zero part $\M{z} \in \texttt{Zeros}$.
    This operation strictly decreases the value of the competitor part \M{r} (breaking the tie) without increasing the value of the donor part \M{z} (which was 0 for agent $i$) by too much.
    Once all ties are resolved, agent $i$ is presented with the partition $M'$ and selects \M{i}.

    \item \textbf{Phase 3: Recursive Completion.}
    As the number of remaining agents decreases, the pool of zero-valued parts may become insufficient to support the necessary swapping operations.
    Therefore, for the final $2D - 1$ agents, we collect all items currently assigned to them and invoke the \textsc{BoundedProp} algorithm.

    Crucially, we configure \textsc{BoundedProp} to use the modified \rr (instead of standard \rr) for its initialization and all internal recursive steps. 
    This ensures that the bounded influence properties are preserved.
\end{itemize}

\ifthenelse{\equal{\algostyle}{algorithm2e}}{
\begin{algorithm}
  \caption{Algorithm for the case where all agents have bounded influence}\label{alg:bounded-influence}
  \SetKwFunction{FBoundedInfluence}{BoundedInfluence}
  \SetKwFunction{FLargeProp}{BoundedProp}
  \SetKwProg{Fn}{Function}{:}{}
  
  \Fn{\FBoundedInfluence{$v_1, \dots, v_n, D$}}{
      \tcp{Phase 1: Initialization via Modified \rr}
      $M \leftarrow \textsc{ModifiedRoundRobin}(v_1, \dots, v_n)$\;
      
      \tcp{Phase 2: Iterative Rebundling}
      \For{$i \leftarrow 1$ \KwTo $n - (2D - 1)$}{
          Let $\texttt{Equals} \gets \{ j > i : v_i(\M{j}) = v_i(\M{i}) \}$\;
          Let $\texttt{Zeros} \gets \{ j > i : v_i(\M{j}) = 0 \}$\;
          \tcp{Note: $|\texttt{Equals}| < D$ and $|\texttt{Zeros}| \ge D-1$}

          \While{$\texttt{Equals} \ne \emptyset$}{
              Pick $r \in \texttt{Equals}$ and $z \in \texttt{Zeros}$\;
              Let $g_{\texttt{best}}$ be the top item for agent $i$ in $\M{r}$\;
              Let $g_{\texttt{zero}}$ be the top item in $\M{z}$ (value 0 for $i$)\;
              Swap $g_{\texttt{best}}$ and $g_{\texttt{zero}}$ between $\M{r}$ and $\M{z}$\;
              Update $\texttt{Equals}$ and $\texttt{Zeros}$\;
          }
          Present the unpicked parts in $M$ to agent $i$\;
      }
      
      \tcp{Phase 3: Recursive Completion}
      Let $a_1, \dots, a_{n'}$ be the remaining $n' \le 2D-1$ agents\;
      $U_{\texttt{next}} \leftarrow \bigcup_{j: \M{j} \text{ is unpicked}} \M{j}$\;
      \tcp{Invoke \FLargeProp using Modified \rr for all steps}
      \FLargeProp{$U_{\texttt{next}}, v_1, \dots, v_{n'}$}
  }
\end{algorithm}
}{}

The following lemma completes the proof of \Cref{thm:bounded-influence-guarantee}.
\begin{restatable}{lemma}{boundedinfluenceguarantee}
    \label{lem:bounded-influence-guarantee}
    For each agent $a$, the algorithm described above yields a part of value at least
    \begin{itemize}
        \item $v_a([m])/D - 2\lceil \log(2D-1) \rceil - 1$, if $n \le 2D-1$ or $a$ is among the last $2D-1$ agents; and 
        \item $v_a([m])/D - 1$, if $a$ is among the first $n-(2D-1)$ agents.
    \end{itemize}
\end{restatable}

\begin{proof}
    We split the proof into two cases depending on the position of $a$.
    
 \paragraph{Case 1: $n > 2D-1$ and $a$ is among the first $n-(2D-1)$ agents.}
    We proceed via induction on $a$ to show that every such agent $a$ precisely picks the part $\M{a}$.

    For the base case, consider $a=1$.
    From \Cref{lem:mod-rr-guarantee}, we know $v_1(\M{1}) \ge v_1([m])/D - 1$.
    However, there might exist some other parts $\M{j}$ ($j>1$) such that $v_1(\M{j}) = v_1(\M{1})$ (strict preference is impossible, by \Cref{lem:mod-rr-guarantee}).
    Let these parts be the set \texttt{Equals}.
    
    Recall that $\M{j}$ has positive value for agent $1$ \emph{only if} $j \in I(1)$. 
    Since $|I(1)| \le D$, there are at most $D-1$ potential equals among the remaining agents.
    Conversely, any agent $k \notin I(1)$ has a part $\M{k}$ valued at exactly 0 by agent $1$.
    
    Since there are at least $2D-1$ agents remaining after agent $1$, the number of potential ``Zero'' parts is at least:
    \[ (n-1) - |I(1)| \ge (2D-1) - D = D-1. \]
    Since we have at least $D-1$ zero parts, we can perform 1-for-1 swaps to devalue the parts in \texttt{Equals}. 
    We pair every Equal part $\M{r}$ with a Zero part $\M{z}$ and swap agent $1$'s favorite item from $\M{r}$ to $\M{z}$.
    \begin{itemize}
        \item The Equal part strictly decreases in value (losing a positive item, gaining a zero-value item). 
        Thus, $r$ is removed from \texttt{Equals}.
        \item The Zero part increases by the value of the transferred item.
        Thus, $z$ is removed from \texttt{Zeros}.
    \end{itemize}
    By our assumption $v_1([m])/D \ge 2 \lceil \log (2D-1) \rceil + 2 \ge 2$, the original part $\M{1}$ has value $\ge 2$. 
    Thus, the newly augmented Zero part (now having value $\le 1$) remains strictly less preferred than $\M{1}$.
    Since we can repeat this logic in such a way that each Zero part receives at most one item, we enforce that $\M{1}$ becomes the unique favorite for agent $1$.
    The agent receives $\M{1}$, preserving the Modified Round-Robin guarantee of $v_1([m])/D - 1$.

    Consider an agent $a > 1$. 
    Assume all prior agents picked their intended parts.
    Crucially, the value of $\M{a}$ has not decreased due to the actions of prior agents.
    Thus, when agent $a$ arrives, the bounds and logic are identical to the base case, ensuring $a$ picks $\M{a}$ which is of value at least $v_a([m])/D - 1$.

    \paragraph{Case 2: $a$ is among the last $2D-1$ agents (or $n < 2D-1$).}
    In this phase, we execute the \texttt{BoundedProp} algorithm (\Cref{alg:large-prop}), but configured to use modified \rr (\Cref{alg:mod-round-robin}) instead of standard \rr for initialization and recursion.
    Since the input to this call is the set of items $U_{\texttt{next}}$, and since these items were originally distributed using modified \rr (in Phase 1), the initialization step of \texttt{BoundedProp} simply reproduces the existing unpicked parts.
    
    This algorithm executes successfully because the partition satisfies the invariant required by it.
    Specifically, each $\M{i}$ contains at least $2 (\lceil \log (2D-1) \rceil + 1)$ many items of non-zero value for agent $i$ (see \Cref{rem:robustness-alg-largeprop}), since $v_i(\M{i}) \ge v_i([m])/D - 1 \ge 2 (\lceil \log (2D-1) \rceil + 1)$.
    
    The analysis of \Cref{thm:large-prop-guarantee} applies directly to the degradation of value during the recursive stages.
    Let $n' = \min(n, 2D-1)$ be the number of agents in this phase.
    By \Cref{clm:agent-stage-largeprop}, the recursion depth is at most $k = \lceil \log n' \rceil + 1$.
    The algorithm guarantees that at any recursive stage $k$, the value of a part decreases by exactly $2(k-1)$ relative to the starting value (due to the items lost in previous stages).
    Since $2(k-1) \le 2\lceil \log n' \rceil$, the final value is at least:
    \[ 
    \left( \frac{v_a([m])}{D} - 1 \right) - 2\lceil \log n' \rceil \ge \frac{v_a([m])}{D} - 2 \lceil \log(2D-1) \rceil - 1. 
    \]
\end{proof}

\section{When all agents have bounded indifference}
\label{sec:no-large-ties}

In this section, we will prove \Cref{thm:no-large-ties-guarantee}.
Here, we consider the special case of valuations where no $t$ items are of the same value to any agent.
Formally, let $t_a$ be the maximum number of non-zero items that have identical values for agent $a$.
We define the instance-wide tie bound as $t \coloneq \max_{a} t_a$.
Towards that, we assume that the proportional share of each agent is at least $\lceil \frac{t+1}{2} \rceil$.

The algorithm works in two phases.
First, we perform a \textit{Greedy Pre-allocation}: we iterate through agents $1, \dots, n$ and assign each agent a ``core'' set of $k = \lceil \frac{t+1}{2} \rceil$ high-value items from the unallocated set.
Second, we run the standard \rr algorithm on the remaining items to complete the bundles.
The ``head start'' provided by the core items ensures that when agent $i$ arrives (in the order $1, \dots, n$), the part $\M{i}$ is sufficiently distinct from the remaining parts $\M{j}$ ($j > i$) to be the unique favorite. 

    The goal is to construct a partition such that when agent $i$ arrives, it finds $\M{i}$ the uniquely best part.
    Towards this, we give every agent a "head start" of high-value items before running a standard \rr on the rest.
    The procedure consists of three phases:
    \begin{itemize}
    \item \textbf{Phase 1: Greedy Pre-allocation.}
    We define a threshold $k = \lceil \frac{t+1}{2} \rceil$ and iterate through the agents from $1$ to $n$.
    In each step, we set aside agent $i$'s top $k$ most preferred items from the set of currently unallocated items $U$. 
    These form the ``core'' of agent $i$'s future part. 
    
    \item \textbf{Phase 2: \rr.}
    We distribute the remaining items in $U$ among the agents using the standard \rr algorithm.
    
    \item \textbf{Phase 3: Agents start arriving.}
    The final partition $M$ is constructed by merging the two partial allocations: for each agent $i$, the part $\M{i}$ consists of their pre-allocated core items plus their share from the \rr phase.
    
    Finally, agents arrive in the order $1, \dots, n$. 
    At each step, agent $i$ is presented with the remaining unpicked parts in $M$ and selects their favorite. 
    The pre-allocation ensures that $\M{i}$ is sufficiently distinct to be the unique favorite.
    \end{itemize}
    
    \ifthenelse{\equal{\algostyle}{algorithm2e}}{
    \begin{algorithm}
      \caption{Algorithm for the case where agents have bounded indifference}\label{alg:no-large-ties}
      \SetKwFunction{FNoLargeTies}{Bounded-Indifference}
      \SetKwProg{Fn}{Function}{:}{}
      
      \Fn{\FNoLargeTies{$[m], v_1, \dots, v_n$}}{
          \tcp{Phase 1: Greedy Pre-allocation}
          $M_1 \gets (\emptyset, \dots, \emptyset)$ \tcp*{Pre-allocation partition}
          $U \gets [m]$ \tcp*{Set of unallocated items}
          $k \gets \lceil \frac{t+1}{2} \rceil$\;
    
          \For{$i \leftarrow 1$ \KwTo $n$}{
              $G_i \leftarrow \text{top } k \text{ items for agent } i \text{ in } U$\;
              $M_1[i] \leftarrow M_1[i] \cup G_i$\;
              $U \gets U \setminus G_i$\;
          }
    
          \tcp{Phase 2: \rr}
          $R \gets \text{\rr}(U, v_1, \dots, v_n)$\;
          Construct partition $M$ where $\M{i} \leftarrow M_1[i] \cup R[i]$ for each $i \in [n]$\;
          
          \tcp{Phase 3: Agents start arriving}
          \For{$i \leftarrow 1$ \KwTo $n$}{
              Present the unpicked parts in $M$ to agent $i$\;
          }
      }
    \end{algorithm}
    }{}
    
    \begin{figure}[h]
        \centering
        \includegraphics[width=0.5\linewidth]{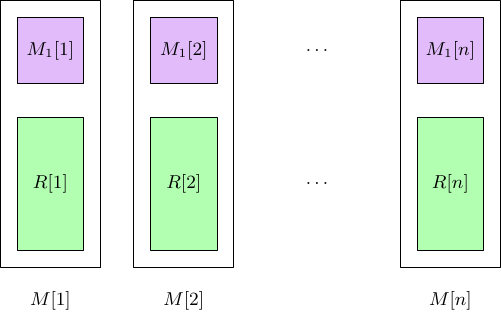}
        \caption{Structure of the partition created in \Cref{alg:no-large-ties}. Each part $\M{i}$ is composed of a core $M_1[i]$ (Phase 1) and a \rr part $R_i$ (Phase 2). The core ensures strict preference over subsequent parts.}
        \label{fig:no-large-ties}
    \end{figure}
    
    \begin{lemma}
        \label{lem:no-large-tie-picks}
        For each $i \in [n]$, agent $i$ picks the part $\M{i}$, which has value at least $\prop_i - \lceil \frac{t+3}{2} \rceil$.
    \end{lemma}
    \begin{proof}
        Let $k = \lceil \frac{t+1}{2} \rceil$.
        First, we show that our assumption on the proportional shares guarantees that every agent can form a core bundle of $k$ items with positive value.
        From the assumption $\prop_i \ge \lceil \frac{t+1}{2} \rceil = k$, and the fact that item values are at most $1$, we have $v_i([m]) \ge nk$.
        Since the maximum value of any single item is $1$, agent $i$ must value at least $nk$ distinct items strictly positively.
        Even in the worst case where all other $n-1$ agents pick items that agent $i$ values positively, at least $nk - (n-1)k = k$ positive items remain available for agent $i$ in the pre-allocation phase.
        Thus, $M_1[i]$ consists entirely of items with strictly positive value for agent $i$.
        
        Next, we establish that agent $i$ strictly prefers their pre-allocated bundle $M_1[i]$ over any bundle $M_1[j]$ created for a subsequent agent $j > i$.
    
        \begin{restatable}{claim}{strictlybetterpart}
            \label{clm:strictly-better-part}
            For any $j > i$, $v_i(M_1[i]) > v_i(M_1[j])$. 
            \Ma{Consequently, $v_i(\M{i})\geq \prop_i - \ceil{ \frac{t+3}{2} }$.
            }
        \end{restatable}
        \begin{proof}
            We note that agent $i$ picks the top $k$ items available in $U$ to form $M_1[i]$. 
            Agent $j$'s bundle $M_1[j]$ is formed from a subset of the remaining items $U \setminus M_1[i]$.
            Since agent $i$ picked greedily, every item in $M_1[i]$ is at least as valuable as every item in $M_1[j]$.
            Thus, $v_i(M_1[i]) \ge v_i(M_1[j])$.
            
            To show strict inequality, assume for contradiction that $v_i(M_1[i]) = v_i(M_1[j])$.
            Since every item in $M_1[i]$ is at least as valuable as every item in $M_1[j]$, the equality implies that all items in $M_1[i]$ and $M_1[j]$ must have the exact same value for agent $i$.
            This would imply a tie of size $|M_1[i]| + |M_1[j]| = 2k$.
            However, $2k = 2\lceil \frac{t+1}{2} \rceil \ge t+1$.
            This contradicts the definition of $t$ (maximum tie size).
            Thus, $v_i(M_1[i]) > v_i(M_1[j])$.

            Next, we consider the allocation of the remaining items.
            By the guarantee of \rr (\Cref{prop:RR-gives-universal-prop}), we know that $v_i(R[i]) \ge v_i(R[j])$ for any $j$.
            Combining this with the previously obtained inequality, we get:
            \[ v_i(\M{i}) = v_i(M_1[i]) + v_i(R[i]) > v_i(M_1[j]) + v_i(R[j]) = v_i(\M{j}). \]
            Therefore, via induction on the agent order, every agent $i$ will find $\M{i}$ to be the uniquely most valuable part among the unpicked parts and will select it.
    
            Finally, we bound the value obtained.
            By the guarantee of \rr (\Cref{prop:RR-gives-universal-prop}), $v_i(R[i]) \ge v_i(U)/n-1$.
            The total value of all pre-allocated bundles is $v_i([m]) - v_i(U)$.
            Since each bundle in $M_1$ has size $k$, and item values are at most 1, the average value of these bundles is at most $k$.
            Specifically, $\frac{1}{n} \sum_{j=1}^n v_i(M_1[j]) \le k = \lceil \frac{t+1}{2} \rceil$.
            Thus, we have 
            \begin{align*}
                v_i(\M{i}) &= v_i(M_1[i]) + v_i(R_i) \\
                &\ge v_i(M_1[i]) + \frac{v_i(U)}{n} - 1 \\
                &= v_i(M_1[i]) + \frac{v_i([m])}{n} - \frac{v_i(M_1[1] \cup \dots \cup M_1[n])}{n} - 1 \\
                &\ge v_i(M_1[i]) + \frac{v_i([m])}{n} - \left\lceil \frac{t+1}{2} \right\rceil - 1\\
                &\ge \prop_i - \left\lceil \frac{t+3}{2} \right\rceil.
            \end{align*}
        \end{proof}
    Combining \Cref{lem:no-large-tie-picks,clm:strictly-better-part}, we have proved \Cref{thm:no-large-ties-guarantee}. 
    \end{proof}
        
    \begin{remark}
        \label{rem:no-large-ties-discussion}
        We note the following analogues to the discussion in \Cref{sec:large-prop}.
        \begin{enumerate}
            \item For an arbitrary fixed ordering, the algorithm guarantees success for all agents $1, \dots, i-1$, where $i$ is the first agent who violates the proportionality requirement.
            \item The requirement that $\prop_i$ is large is used primarily to ensure that agent $i$ can find $k = \lceil \frac{t+1}{2} \rceil$ items of positive value in Phase 1 after the previous $i-1$ agents have picked.
            A strictly weaker sufficient condition is that agent $i$ (who arrives $i$-th in the sequence) simply values at least $i \cdot k$ items in the instance strictly positively.
            This ensures that even if the preceding $i-1$ agents pick $k$ items each, agent $i$ still finds $k$ positive items to form their core bundle.
        \end{enumerate}
    \end{remark}
    
    \subsection{Guarantees for general instances}
    
    \Cref{thm:no-large-ties-guarantee} assumed a fixed arrival order in which agents have proportional share at least $\lceil \frac{t+1}{2} \rceil$. 
    In the same vein as \Cref{thm:large-prop-guarantee-extended}, we now show that for any instance, we can compute an arrival order that generalizes this result, ensuring the bound applies to all agents regardless of their proportional share.
    In \Cref{thm:large-prop-guarantee-extended}, generalizing the result to arbitrary instances required a complex procedure involving candidate sets. 
    Here, however, the structure of the algorithm allows for a much simpler strategy: sorting agents by their proportional share is sufficient to extend the guarantee to all. 
    
    \begin{theorem}
        \label{thm:no-large-ties-guarantee-extended}
        For any fair division instance, there exists an arrival order $\tau$ of the agents and an algorithm that guarantees each agent $\tau_i$ a part of value at least:
        \[ 
            \prop_{\tau_i} - \left\lceil \frac{t+3}{2} \right\rceil.
        \]
        Moreover, the computation of $\tau$ and the execution of the algorithm can be done in polynomial time.
    \end{theorem}
    
    \begin{proof}
        Let $\tau$ be the order of agents obtained by sorting them in decreasing order of their proportional shares, i.e., $v_{\tau_1}([m]) \ge v_{\tau_2}([m]) \ge \dots \ge v_{\tau_n}([m])$.
    
        Recall that $k=\lceil \frac{t+1}{2} \rceil$ is the threshold used in \Cref{thm:no-large-ties-guarantee}.
        If all agents have proportional shares at least $k$, then by \Cref{thm:no-large-ties-guarantee} we are done.
        Otherwise, we identify the first agent, say $\tau_b$, whose proportional share falls below the threshold. 
        That is, $\prop_{\tau_i} \ge k$ for all $i < b$, and $\prop_{\tau_i} < k$ for all $i \ge b$.
    
        By \Cref{rem:no-large-ties-discussion}, \Cref{alg:no-large-ties} executes successfully for all agents $\tau_i$ with $i \le b$, guaranteeing each such agent $\tau_i$ a part of value at least $\prop_{\tau_i} - \lceil \frac{t+3}{2} \rceil$.
    
        For any remaining agent $i \ge b$, we have $\prop_{\tau_i} \le \prop_{\tau_b} < k < \lceil \frac{t+3}{2} \rceil$.
        Consequently, the guarantee $\prop_{\tau_i} - \lceil \frac{t+3}{2} \rceil$ is strictly negative.
        Since valuations are non-negative, any part received by these agents trivially satisfies the guarantee.
    \end{proof}

\section{Ordered additive valuations and relaxations}
\label{sec:ordered-additive-relaxations}

In this section we will prove \Cref{thm:ordered-additive-relaxation-guarantees}.
We begin our discussion with ordered additive instances and then move to the relaxations.
For the latter we split the results based on the structure such as: when swaps form a linearly-separable family, a laminar family, or are adjacent.
We conclude our discussions about adjacent swaps with some observations pertaining to improved guarantees when the valuation functions are $\delta$-Lipschitz.

We begin by noting that if an agent $a$ has ties in their valuation (i.e., $v_a(g_1) = v_a(g_2)$ for distinct items $g_1, g_2$), there are multiple permutations that derive $v_a$ (as defined in \Cref{subsec:contribution}).
However, it is computationally easy to find \emph{a} permutation $\sigma_a$ that derives $v_a$ simply by sorting the items based on value, breaking ties arbitrarily.
Thus, throughout this paper, we assume that we are given specific permutations $\sigma_1, \dots, \sigma_n$ that derive the agent valuations.

We study instances where there is a master list $\pi$ that is ``close'' to each of $\sigma_1, \dots, \sigma_n$, and we obtain fair partitions whose quality depends on this distance.
However, we note that a different choice of tie-breaking for agents' permutations might result in a master list that is even closer to the set of permutations.
We leave the optimization problem of finding the specific tie-breaking permutations that minimize the distance to a master list for future work.

Given an $m$-length permutation $\sigma$ and two positions $i,j \in [m]$ with $i < j$, a \emph{swap} operation $(i,j)$ transforms $\sigma$ into a new permutation $\sigma'$ by exchanging the items at positions $i$ and $j$, leaving all other positions unchanged.
If $j=i+1$, the swap is called \emph{adjacent}.
See \Cref{fig:types-of-swaps} for an illustration.
We analyze instances where each agent's permutation is at most a certain distance $d$ (measured in the number of swaps) from the master list $\pi$.

Throughout this section, we assume without loss of generality that the master list is the identity permutation, i.e., $\pi = \langle 1, \dots, m \rangle$.
This is justified because the items can always be relabeled to match the order of the master list.

In our proofs, we will repeatedly use a valuation function which is derived from the master list $\pi$, defined as follows.

\begin{definition}[Projected valuation]
\label{def:proj-val}
Let the master list be $\pi = \langle 1, \dots, m \rangle$.
For any agent $a$, let $x_1 \ge x_2 \ge \dots \ge x_m$ be the set of values $\{v_a(g) : g \in [m]\}$ sorted in non-increasing order.
We define the \emph{projected valuation} $\hat{v}_a$ such that $\hat{v}_a(i) = x_i$ for all $i \in [m]$.
Note that by definition, $\hat{v}_a$ is derived from the master list $\pi$.
\end{definition}

\paragraph{Related work}
The problem of computing a master list that is at distance at most $d$ from each agent's permutation has been studied.
Popov~\cite{DBLP:journals/tcs/Popov07} showed that this problem is NP-hard for both arbitrary and adjacent swap distances.
Cunha, Sau, and Souza~\cite{DBLP:journals/almob/CunhaSS24} focused on arbitrary swap distance and showed that it is fixed-parameter-tractable parameterized by $d$.
We note that this result can be extended to the case of adjacent swap distance as well.
Schlotter~\cite{DBLP:journals/tcs/Schlotter24} studied the related problem of determining if a preference system is close to admitting a master list.
Specifically, they studied the question of whether there are $d$ adjacent swaps that can be made in total so that all the agents agree on a master list. 
We note that their model permits agents to have weak partial orders, and defines the notion of distance accordingly.

\subsection{Ordered additive instances}
First, we show that when the instance is strictly ordered additive (i.e., when $d=0$), we can find a solution,\hide{ where all \parts have a value of at least $\prop_a - 1$ for each agent $a$,} using the standard \rr protocol.

\begin{restatable}{lemma}{orderedadditive}
\label{lem:ordered-additive}
In an ordered-additive instance, there is a partition $M=(\M{1}, \dots, \M{n})$ such that $v_a(\M{j}) \ge \prop_a - 1$ for each agent $a$ and each \part $\M{j}$. 
Moreover, such a partition can be computed in polynomial time.    
\end{restatable}

    
\begin{proof}
    Let $\pi \in S_m$ be the master list from which all agents' preferences are derived.
    Let $M=(\M{1}, \ldots, \M{n})$ be the partition output by \rr on the $n$ agents.
    Since all valuations are derived from the same list $\pi$, the most preferred item is identical for every agent at every step of the algorithm.
    Consequently, the resulting partition $M$ depends only on the master list $\pi$ and is independent of the specific values given by the agents.
    Thus, it follows that the same partition $M$ would have emerged even if the agent valuations were identical; specifically, even if \rr was executed using the valuation function $v_a$ for any specific agent $a$. 

    Fix an arbitrary agent $a$.
    The partition $M$ obtained is identical to the partition that would result if \emph{all} $n$ agents had the valuation function $v_a$.
    To prove $v_a(\M{j}) \geq \prop_a-1$ for each agent $a$ and each \part $\M{j}$, it is sufficient to show that $v_a(\M{n}) \geq \prop_a-1$. 
    This is because, from \Cref{prop:RR-gives-universal-prop}, we know that for any $j \leq n$, $v_a(\M{j}) \geq v_a(\M{n})$. 

    Recall from \Cref{prop:RR-gives-universal-prop} that \rr produces a partition in which the $n$-th \part obtained using the valuation function $v_a$ satisfies $v_a(\M{n}) \geq \prop_a -1$.
    We conclude by noting that the partition $M$ is computable in polynomial time, and the property $v_a(\M{j}) \ge v_a(\M{n})$, for each $j \in [n]$, gives the desired guarantee.
\end{proof}

We now turn to relaxations where the valuations deviate from the master list.
We distinguish between swap types—arbitrary or adjacent—and their structural forms, such as linearly-separable or laminar families, to obtain fine-grained guarantees.

\subsection{Arbitrary swaps}
We now consider the case where an agent's true valuation is derived from a permutation obtained by applying at most $k$ arbitrary swaps to the master list.

\begin{restatable}{lemma}{arbitraryswaps}
\label{lem:arbitrary-swaps}
    In an instance where each agent $a$'s permutation is at most $k_a$ arbitrary swaps away from the master list $\pi$, there exists a partition $M=(\M{1}, \dots, \M{n})$ such that $v_a(\M{j}) \ge \prop_a - 1 - k_a$, for each agent $a$ and \part $\M{j}$. 
    Moreover, such a partition can be computed in polynomial time.
\end{restatable}



\begin{proof}
    Let $\hat{v}_a$ be the \emph{projected valuation} of agent $a$ (derived strictly from $\pi$, as in \Cref{def:proj-val}).
    We run \rr using the projected valuations $\{\hat{v}_1, \dots, \hat{v}_n\}$ to obtain a partition $M$.
    Due to \Cref{lem:ordered-additive}, we have 
    $\hat{v}_a(\M{j}) \ge \prop_a - 1$ for each agent $a$ and \part $\M{j}$. 

    For each item $g$, consider the change in value of good $g$ for agent $a$, defined as $\change(g, \hat{v}_a, v_a) \coloneq \hat{v}_a(g) - v_a(g)$. 
    We note that a positive value indicates a fall in value for good $g$ when moving from the projected valuation $\hat{v}_a$ to the true valuation $v_a$. 

    The value of $\M{j}$ is given by:
    \[
        v_a(\M{j}) = \sum_{g \in \M{j}} v_a(g) = \sum_{g \in \M{j}} \left( \hat{v}_a(g) - \change(g, \hat{v}_a, v_a) \right) = \hat{v}_a(\M{j}) - \sum_{g \in \M{j}} \change(g, \hat{v}_a, v_a). 
    \]
    An arbitrary swap of items $g_1, g_2$ affects only the values of $g_1$ and $g_2$.
    Since the agent's permutation $\sigma_a$ is at most $k_a$ swaps away from $\pi$, there are at most $k_a$ items in $\M{j}$ that have a positive change in value (each bounded by 1 since item values are in $[0,1]$).
    Hence, we conclude that $v_a(\M{j}) \ge \prop_a - 1 - k_a$.
    The valuations $\{\hat{v}_1, \dots, \hat{v}_n\}$ and the partition $M$ can be computed in polynomial time, and thus we have proved the lemma.
\end{proof}

\subsection{Linearly separable swaps}
\label{subsec:linearly-separable-swaps}

We say that a set of swaps $S$ is \emph{linearly separable} if for any two swaps $(s_1, s_2), (s_3,s_4) \in S$ with $s_1 \le s_3$, it implies that $s_2 \le s_3$ and $s_1 < s_4$ (i.e., the intervals defined by the swaps are internally disjoint). 
See \Cref{fig:types-of-swaps} for an illustration.
We now consider the case where each agent is away from the master list by a linearly separable set of swaps.

\begin{restatable}{lemma}{linearlyseparableswaps}
\label{lem:linearly-separable-swaps}
Suppose that an agent $a$'s valuation is derived from a permutation that is away from the master list $\pi$ by a set of linearly separable swaps.
Let $M=(\M{1}, \dots, \M{n})$ be the output of \rr on the master list.
Then, $v_a(\M{j}) \ge \prop_a - 2$, for each part $\M{j}$. 
\end{restatable}



\begin{proof}
    We proceed similarly to the previous lemmas. 
    Let $\hat{v}_a$ be the projected valuation of agent $a$.
    We run \rr using the projected valuations $\{\hat{v}_1, \dots, \hat{v}_n\}$ and obtain a partition $M$.
    Due to \Cref{lem:ordered-additive}, we have:
    \begin{align*}
        v_a(\M{j}) 
        &= \sum_{g \in \M{j}} v_a(g) = \sum_{g \in \M{j}} \left( \hat{v}_a(g) - \change(g, \hat{v}_a, v_a) \right) \\
        &= \hat{v}_a(\M{j}) - \sum_{g \in \M{j}} \change(g, \hat{v}_a, v_a) \\
        &\ge \left( \prop_a - 1 \right) - \sum_{g \in \M{j}} \change(g, \hat{v}_a, v_a).
    \end{align*}
    
    Consider any \part $\M{j}$. 
    Consider the total potential loss from a set of linearly separable swaps $S$. 
    For a single swap $(p, q) \in S$ where $p < q$, the item originally at $p$ moves to $q$, incurring a loss of $\hat{v}_a(p) - \hat{v}_a(q)$.
    Summing over all swaps in $S$:
    \[
    \sum_{(p,q) \in S} (\hat{v}_a(p) - \hat{v}_a(q)) \le \sum_{i=1}^{m-1} (\hat{v}_a(i) - \hat{v}_a(i+1)) = \hat{v}_a(1) - \hat{v}_a(m) \le 1.
    \]
    The inequality holds because the intervals $[p,q]$ for different swaps are internally disjoint, so the sum is bounded by the telescoping sum over the entire domain $[m]$.
    Thus, for the \part $\M{j}$, we have $v_a(\M{j}) \geq \prop_a-1-1 = \prop_a-2$.
    
    Since the valuations $\hat{v}_a$ and consequently the partition $M$ can be computed in polynomial time, we have proved the lemma.
\end{proof}

We now generalize \Cref{lem:linearly-separable-swaps} to study the scenario where a set of swaps can be partitioned into multiple linearly separable subsets.

\begin{lemma}
    \label{lem:lin-sep-guarantee}
    Let $S$ be a set of swaps performed to obtain the permutation $\sigma_a$ from the master list.
    Let $S_1, \dots, S_t$ be a partition of $S$ such that each subset $S_i$ is linearly separable.
    Let $M=(\M{1}, \dots, \M{n})$ be the output of \rr on the master list.
    Then, $v_a(\M{j}) \ge \prop_a - 1 - t$ for each \part $\M{j}$.
\end{lemma}

\begin{proof}
    Let $\hat{v}_a$ be the projected valuation function of agent $a$. 
    From \Cref{lem:ordered-additive}, we know $\hat{v}_a(\M{j}) \ge \prop_a - 1$. 
    We construct a sequence of valuations $v_a^0, \dots, v_a^t$ where $v_a^0 = \hat{v}_a$, and $v_a^i$ is obtained from $v_a^{i-1}$ by applying the set of swaps $S_i$.

    For arbitrary valuation functions $v_1, v_2:[m] \to [0,1]$, let $\change(g,v_1, v_2) \coloneq \max( 0, v_1(g) - v_2(g))$ denote the magnitude of decrease in value of good $g$ when moving from $v_1$ to $v_2$.
    We observe the following property regarding the total change, which completes the proof.

    \begin{restatable}{claim-inside-lemma}{decomposingchange}
        \label{clm:decomposing-change}
        $\sum_{g \in [m]} \change(g,v^0_a, v_a) \le \sum_{g \in [m]} \sum_{i=1}^t \change(g,v^{i-1}_a, v^i_a)$. 
        Moreover, this implies $v_a(\M{j}) \geq \prop_a - 1 - t$. 
        \qedhere
    \end{restatable}

    \begin{proof}
        Any decrease in value for an item occurs solely due to a swap.
        The total loss incurred by the swaps is bounded by the sum of value differences at the swap positions.
        Since the subsets $S_1, \dots, S_t$ partition $S$, we can decompose the total loss as follows:
        \begin{align*}
            \sum_{g \in [m]} \change(g, v^0_a, v_a) 
            &\le \sum_{(p_1, p_2) \in S} \left( \hat{v}_a(p_1) - \hat{v}_a(p_2) \right) \\
            &= \sum_{i=1}^t \sum_{(p_1,p_2) \in S_i} \left( \hat{v}_a(p_1)-\hat{v}_a(p_2) \right) \\
            &= \sum_{i=1}^t \sum_{g \in [m]} \change(g,v^{i-1}_a, v^i_a).
        \end{align*}
        
        Using this decomposition, we lower bound the value of bundle $\M{j}$: 
        \begin{align*}
            v_a(\M{j}) &= \sum_{g \in \M{j}} v_a(g) \\ 
            &\ge \sum_{g \in \M{j}} \left( \hat{v}_a(g) - \change(g,v^0_a, v_a) \right) \\ 
            &\ge (\prop_a - 1) - \sum_{g \in [m]} \change(g,v^0_a, v_a) \\
            &\ge \prop_a - 1 - \sum_{i=1}^t \sum_{g \in [m]} \change(g,v^{i-1}_a, v^i_a) \tag{by \Cref{clm:decomposing-change}} \\
            &\ge \prop_a - 1 - \sum_{i=1}^t 1 \tag{since each $S_i$ is linearly separable} \\
            &= \prop_a - 1 - t.
        \end{align*}
    \end{proof}
\end{proof}

Extending the above results, we can infer the following general bound for the entire instance.

\begin{restatable}{lemma}{lemlinsepguarantee2}
    \label{lem:lin-sep-guarantee2}
    For each agent $a$, let $t_a$ be the minimum integer such that the agent's permutation can be obtained from the master list via a set of swaps that can be partitioned into $t_a$ linearly separable subsets.
    Then, there is a partition $M = (\M{1}, \dots, \M{n})$ such that for each agent $a$ and \part $\M{j}$:
    \[
        v_a(\M{j}) \ge \prop_a - 1 - t_a
    \]
    Moreover, such a partition can be computed in polynomial time.
\end{restatable}

\subsection{Adjacent swaps}
\label{subsec:adjacent-swaps}
We now consider the case where each agent is at most $k$ adjacent swaps away from the master list.

\begin{restatable}{lemma}{adjacentswaps}
\label{lem:adjacent-swaps}
    In an instance where each agent $a$'s permutation is at most $k_a$ adjacent swaps away from the master list $\pi$, there exists a partition $M=(\M{1}, \dots, \M{n})$ such that $v_a(\M{j}) \ge \prop_a - \sqrt{2k_a} - 1$, for each agent $a$ and part $\M{j}$. 
    Moreover, such a partition can be computed in polynomial time.
\end{restatable}

    

\begin{proof}
    Let $M=(\M{1}, \dots, \M{n})$ be the partition output by \rr using the projected valuations $\{\hat{v}_1, \dots, \hat{v}_n\}$.
    
    Fix an agent $a$. 
    Let $S$ denote a minimum-size multiset of adjacent swaps applied to the master list to obtain $\pi_a$. 
    By definition, $|S| \le k_a$.
    The following procedure describes such an $S$:
    \begin{enumerate}
        \item Initialize $S = \emptyset$.
        \item\label{step:adjacent-swaps}
        Let $g_\ell$ be the item present in position $p_1$ in the current permutation and $p_2$ in the target $\pi_a$, with $p_1 > p_2$. 
        We choose $g_\ell$ such that it has the least value of $p_2$ among all such items (i.e., we fix the leftmost ``incorrect'' position).
        We add to $S$ the sequence of swaps: $(p_1, p_1-1), (p_1-1, p_1-2), \dots, (p_2+1, p_2)$.
        \item Apply these swaps to the current permutation and repeat Step 2 until the permutation matches $\pi_a$.
    \end{enumerate}

    Let $S = S_1 \cup \dots \cup S_t$ be a partition of $S$ into the minimum number of subsets such that each subset $S_i$ is linearly separable.
    Note that $t \le |S| \le k_a$ since a trivial partition exists where each $S_i$ contains a single swap.
    We now claim that we can partition $S$ into significantly fewer sets.

    \begin{claim}
        \label{clm:upperbounding-t}
        $t \le \sqrt{2k_a}$.
    \end{claim}
    \begin{proof}
        For any two permutations $\pi_1$ and $\pi_2$, the minimum number of adjacent swaps to reach $\pi_2$ from $\pi_1$ is equal to the number of pairs $\{g_1, g_2\}$ of items such that $g_1$ precedes $g_2$ in $\pi_1$ and $g_2$ precedes $g_1$ in $\pi_2$ \cite{DBLP:journals/tcs/Jerrum85}.

        If each swap is present at most $\sqrt{2k_a}$ times in $S$, then we are done because we can define \Ma{$S_i$ as the set containing the $i^{th}$ copy of any swap} that is present in $S$.
        
        Suppose that some swap $(p,p-1)$ is present $\sqrt{2k_a}+1$ times.
        Then, this means that Step \ref{step:adjacent-swaps} was called $\sqrt{2k_a}+1$ times with the swap $(p, p - 1)$ being added.
        Observe that the value of $p_2$ (which is at most $p$) is different for each such addition. 
        This means that the total number of swaps is at least $(\sqrt{2k_a})(\sqrt{2k_a}+1)/2 >k_a$, and thus we have a contradiction to the fact that the permutation is only $k_a$ swaps away from the master list.
    \end{proof}

    Combining \Cref{clm:upperbounding-t} with the guarantee from \Cref{lem:lin-sep-guarantee}, we obtain the result.
\end{proof}

\subsection{$\delta$-Lipschitz valuations}
\label{subsec:Lipschitz}

We now analyze the performance when the valuation function $v_a$ is $\delta$-Lipschitz with respect to the master list $\pi$. 
Formally, we assume $|v_a(\pi(i)) - v_a(\pi(i+1))| \le \delta$ for each $i \in [m-1]$.
Equivalently, the values are given by a sequence $1 \ge \delta^a_1 \ge \dots \ge \delta^a_m \ge 0$ such that $v_a(\pi(i)) = \delta^a_i$ and $|\delta^a_i - \delta^a_{i+1}| \le \delta$. 

\begin{restatable}[\app]{lemma}{lemlipschitz}
\label{lem:lipschitz}
    In an instance where each agent's valuation is $\delta$-Lipschitz and is derived from a permutation that is away from the master list $\pi$ by at most $k$ adjacent swaps, there is a partition $M=(\M{1}, \dots, \M{n})$ such that $v_a(\M{j}) \ge \prop_a - 1 - \delta k$, for each agent $a$ and each \part $\M{j}$.
    Moreover, such a partition can be computed in polynomial time.
\end{restatable}

\begin{proof}
    Let $M$ be the \rr partition on the projected valuations, defined in \Cref{def:proj-val}.
    Recall that $v_a(\M{j}) \ge (\prop_a - 1) - \text{loss}$, where the loss is the sum of value differences caused by the swaps.
    For a single adjacent swap at position $i$ (swapping items at ranks $i$ and $i+1$), the loss incurred is $v_a(\pi(i)) - v_a(\pi(i+1))$.
    By the $\delta$-Lipschitz property, this difference is at most $\delta$.
    Since there are at most $k$ such swaps, the total loss is at most $\sum_{s=1}^k \delta = \delta k$.
    Thus, $v_a(\M{j}) \ge \prop_a - 1 - \delta k$.
\end{proof}
\subsection{Laminar swaps}
\label{subsec:Laminar}

We extend our analysis to the case of arbitrary swaps that possess the \emph{laminar} property to obtain a better bound than the one given by \Cref{lem:arbitrary-swaps}.
A multiset of swaps $S$ is \emph{laminar} if for any pair of distinct swaps, they are either internally disjoint or one is contained inside the other.
Formally, for $s = (a, b)$ and $s'=(c,d) \in S$ (with $a<b$ and $c<d$), either $b \le c$ ($s$ and $s'$ are disjoint), or $a \le c$ and $d \le b$ ($s$ contains $s'$), or $c \le a$ and $b \le d$ ($s'$ contains $s$). 
If $s$ contains $s'$ but $s \ne s'$, then we say that $s$ strictly contains $s'$.

\begin{definition}[Depth of a laminar family of swaps]
    \label{def:depth-laminar-set}
    We define the \emph{depth} of a laminar family $S$, denoted by $\text{depth}(S)$, via a recursive decomposition.
    \begin{enumerate}[itemsep=0pt]
    \item Let $S_1$ be the set of \emph{maximal} swaps in $S$. A swap $s \in S$ is maximal if it is not strictly contained in any other swap in $S$. 
    (If there are multiple copies of a maximal swap, we pick exactly one copy for $S_1$).
    \item Remove $S_1$ from $S$.
    \item Repeat the process to find $S_2, S_3, \dots, S_t$ until $S$ is empty.
    \end{enumerate}
    Then, $\text{depth}(S) \coloneq t$.
\end{definition}

We note that linearly separable swaps correspond to a laminar family of depth 1.

\begin{lemma}
\label{lem:laminar-swaps-guarantee}
    Suppose that each agent $a$'s valuation is derived from a permutation that is away from the master list $\pi$ by a family $S^a$ of laminar swaps.
    Then, there is a partition $M=(\M{1}, \dots, \M{n})$ such that for each agent $a$ and \part $\M{j}$: $v_a(\M{j}) \ge \prop_a - 1- \text{depth}(S^a)$.
\end{lemma}

\begin{proof}
    We show that the decomposition $S^a_1, \dots, S^a_{\text{depth}(S^a)}$ described in \Cref{def:depth-laminar-set} partitions $S^a$ into linearly separable sets.
    Consider any set $S^a_i$.
    By the property of $S^a$, any two swaps in $S^a_i$ are either internally disjoint or one contains another.
    By construction, $S^a_i$ does not contain two copies of the same swap, and there is no swap that strictly contains another swap (otherwise the contained swap would not be maximal in the current set).
    Thus, all swaps in $S^a_i$ are pairwise internally disjoint.
    
    A set of pairwise internally disjoint swaps is linearly separable.
    Thus, $S^a_1, \dots, S^a_{\text{depth}(S^a)}$ forms a partition of $S^a$ into linearly separable sets.
    Applying \Cref{lem:lin-sep-guarantee} with $t_a = \text{depth}(S^a)$ for each agent $a$, we obtain the desired guarantee.
\end{proof}

By combining \Cref{lem:ordered-additive,lem:lin-sep-guarantee,lem:arbitrary-swaps,lem:adjacent-swaps,lem:laminar-swaps-guarantee}, \Cref{thm:ordered-additive-relaxation-guarantees} is proved.

\section{Impossibility of a better universal guarantee}
\label{sec:one-sided-lower-bound}
In this section we prove \Cref{thm:lower-bound}.
Recall that via the discrepancy argument, we can obtain a partition where all parts are of value at least $\prop_a - \bigoh(\sqrt{n})$ for each agent $a$.
In this section, we show that this guarantee is tight.
Specifically, we construct instances where, for \emph{any} partition, there exists an agent and a part such that the value of the part for the agent is ``low'' compared to their proportional share.
We note that our construction and proof builds upon the work of Manurangsi and Meka \cite{doi:10.1137/1.9781611978964.20Manurangsi26}.

\subsubsection*{Construction and Proof Overview}
We construct the valuation matrix $\mathbf{B}$ using Hadamard matrices.
Let $\mathbf{H}$ be an $\frac{n}{2} \times \frac{n}{2}$ Hadamard matrix, and $\mathbf{J}$ be the $\frac{n}{2} \times \frac{n}{2}$ all-ones matrix.
Note that the entries of a Hadamard matrix are in $\{+1, -1\}$, and the scalar product of any two distinct columns is zero.
We also assume that $\mathbf{H}$ is constructed in such a way that it contains a row of all ones \cite{hedayat1978hadamard}.
We define the matrices $\mathbf{A}$ and $\mathbf{A}_{\rm opp}$ as:
\[
\mathbf{A} \coloneq \frac{1}{2}(\mathbf{J} + \mathbf{H}), \quad \mathbf{A}_{\rm opp} \coloneq \frac{1}{2}(\mathbf{J} - \mathbf{H})
\]
We then construct a block matrix $\mathbf{\tilde{A}}$ of size $n \times n$ as:
\[
\mathbf{\tilde{A}} = \begin{pmatrix} \mathbf{A} & \mathbf{A} \\ \mathbf{A}_{\rm opp} & \mathbf{A}_{\rm opp} \end{pmatrix}.
\]

Let $t = n/4$.
The fair division instance consists of $n$ agents and $m=nt$ items.
The valuation matrix $\mathbf{B} \in \{0,1\}^{n \times nt}$ is constructed by horizontally stacking the matrix $\mathbf{\tilde{A}}$ exactly $t$ times:
\[
\mathbf{B} \coloneq [\underbrace{\mathbf{\tilde{A}} ~|~ \mathbf{\tilde{A}} ~|~ \cdots ~|~ \mathbf{\tilde{A}}}_{t \text{ times}}].
\]

In this instance, the rows, the columns, and the entry $\mathbf{B}[a,g]$ correspond to the $n$ agents, $nt$ items, and the value agent $a$ assigns to item $g$, respectively.
We observe the following due to the structure of Hadamard matrix $\mathbf{H}$.

\begin{restatable}{lemma}{propshareininstance}
    \label{lem:prop-share-in-instance}
    In the constructed instance, the proportional share of any agent takes one of three values: $m/n$ (for the all-ones agent), $0$ (for the all-zeros agent), or $m/(2n)$ (for all other agents).
\end{restatable} 
\begin{proof} 
    Recall that $\mathbf{A} = \frac{1}{2}(\mathbf{J} + \mathbf{H})$ and $\mathbf{A}_{\rm opp} = \frac{1}{2}(\mathbf{J} - \mathbf{H})$, where $\mathbf{H}$ contains a row of all ones. 
    Without loss of generality, suppose that the first row of $\mathbf{H}$ is all ones.
    Then, the first row of $\mathbf{A}$ is all ones, and the first row of $\mathbf{A}_{\rm opp}$ is all zeros.
    By the construction of $\mathbf{\tilde{A}}$ and $\mathbf{B}$, there is an agent (corresponding to the first row of $\mathbf{A}$ in $\mathbf{B}$) who values every item at $1$, and thus has a proportional share of $m/n$.
    Analogously, there is an agent (corresponding to the first row of $\mathbf{A}_{\rm opp}$ in $\mathbf{B}$) who values every item at $0$, and thus has a proportional share of $0$.
    
    For any row $i \ne 1$ of the Hadamard matrix $\mathbf{H}$, the scalar product property (when applied with the first row of ones) implies that exactly $n/4$ entries are $+1$ and $n/4$ entries are $-1$.
    Therefore, the corresponding rows in $\mathbf{A}$ and $\mathbf{A}_{\rm opp}$ contain exactly $n/4$ ones and $n/4$ zeros.
    By the construction of $\mathbf{\tilde{A}}$ and $\mathbf{B}$, the corresponding agent values exactly $(n/4) \cdot 2 \cdot t = m/2$ items at $1$, and exactly $m/2$ items at $0$.
    Thus, the proportional share of such an agent is $m/(2n)$.
\end{proof}

Let $P=(P_1, \dots, P_n)$ be an arbitrary $n$-partition of the items.
Consider any arbitrary \part, denoted by $P_j$.
We will study the value of $P_j$ for each agent.
Let $\chi(P_j) \in \{0,1\}^{nt}$ be the characteristic vector of $P_j$ (here the $i$-th entry is 1 if item $i$ is in $P_j$, and $0$ otherwise).
The values derived by each of the agents from $P_j$ is given by the $n$-dimensional vector $\mathbf{B}\chi(P_j)$.
The deviation from proportional share, for each agent for the \part $P_j$, is given by the vector $\mathbf{z} = \mathbf{B}(1/n \cdot \mathbf{1} - \chi(P_j))$.
The main technical step in this analysis is to prove the following. 

\begin{restatable}{lemma}{weighteddiscbound}
    \label{lem:weighted-disc-bound}
    For any vector $\mathbf{x} \in \{0,1\}^{nt}$, we have $\lVert \mathbf{B} ( 1/n \cdot \mathbf{1} - \mathbf{x}) \rVert_\infty \ge \sqrt{n/32}$.
   \hide{ \[
    \lVert \mathbf{B} ( 1/n \cdot \mathbf{1} - \mathbf{x}) \rVert_\infty \ge \sqrt{n/32}.
    \]}
\end{restatable}

\begin{proof}
\begingroup 
\expandafter\renewcommand\csname theclaim-inside-lemma\endcsname{\ref{lem:weighted-disc-bound}.\arabic{claim-inside-lemma}}
\setcounter{claim-inside-lemma}{0}

We note that $\|\mathbf{v}\|_\infty \ge \frac{1}{\sqrt{n}} \|\mathbf{v}\|_2$, for any $v\in \{0,1\}^{nt}$. Thus, we need to lower bound \hide{ the $\ell_2$ norm of} $\lVert\mathbf{B}(1/n \cdot \mathbf{1} - \mathbf{x})\rVert_2$.
Recall that $t = n/4$ and $\mathbf{B} \coloneq [\mathbf{\tilde{A}} ~|~ \cdots ~|~ \mathbf{\tilde{A}}]$.
    Consider any $\mathbf{x} \in \{0,1\}^{nt}$ and let $p \coloneq 1/n$.
    Let $\mathbf{z} \in \{0,\dots,t\}^n$ be defined by $z_i = \sum_{j=0}^{t-1} x_{i+nj}$.
    Thus, $\mathbf{B}(p \cdot \mathbf{1} - \mathbf{x}) = \mathbf{\tilde{A}}(pt \cdot \mathbf{1} - \mathbf{z})$.
   
    In order to lower bound $\lVert\mathbf{B}(1/n \cdot \mathbf{1} - \mathbf{x})\rVert_2$, we will analyze the $\ell_2$-norm of matrix-vector products involving the block matrix $\mathbf{\tilde{A}}$.
    
    \begin{claim-inside-lemma}
        \label{clm:square-of-l2-norm}
        For any vector $\mathbf{y} \in \mathbb{R}^n$,
        \[
        \lVert \mathbf{\tilde{A}}\mathbf{y} \rVert_2^2 = \frac{n}{4} \left( \sum_{i=1}^n y_i \right)^2 + \frac{n}{4} \sum_{i=1}^{n/2} \left( y_i + y_{i+n/2} \right)^2.
        \]
    \end{claim-inside-lemma}
        \begin{proof}
        Expanding the LHS, we have:
        \begin{align*}
            \|\mathbf{\tilde{A}}\mathbf{x}\|_2^2 &= (\mathbf{\tilde{A}}\mathbf{x})^\top (\mathbf{\tilde{A}}\mathbf{x}) \\
            &= \mathbf{x}^\top \mathbf{\tilde{A}}^\top \mathbf{\tilde{A}} \mathbf{x} \\
            &= \sum_{i=1}^n \sum_{j=1}^n x_i x_j \left( \mathbf{\tilde{A}}_{(\cdot, i)} \cdot \mathbf{\tilde{A}}_{(\cdot, j)} \right),
        \end{align*}
        where $\mathbf{\tilde{A}}_{(\cdot, k)}$ denotes the $k$-th column of $\mathbf{\tilde{A}}$.
    
        Let $\mathbf{Q}_k$ denote the $k$-th column of the matrix $\mathbf{\tilde{A}}$. 
        Expanding the double sum by splitting indices into the left half ($1, \dots, n/2$) and right half ($n/2+1, \dots, n$), we have:
        \begin{align*}
            \|\mathbf{\tilde{A}}\mathbf{x}\|_2^2 &= \sum_{i=1}^{n/2} \sum_{j=1}^{n/2} x_i x_j (\mathbf{Q}_i \cdot \mathbf{Q}_j) + \sum_{i=1}^{n/2} \sum_{j=n/2+1}^{n} x_i x_j (\mathbf{Q}_i \cdot \mathbf{Q}_j) \\
            &\quad\quad + \sum_{i=n/2+1}^{n} \sum_{j=1}^{n/2} x_i x_j (\mathbf{Q}_i \cdot \mathbf{Q}_j) + \sum_{i=n/2+1}^{n} \sum_{j=n/2+1}^{n} x_i x_j (\mathbf{Q}_i \cdot \mathbf{Q}_j).
        \end{align*}
    
        Note that the structure of $\mathbf{\tilde{A}}$ implies that the first $n/2$ columns are identical to the second $n/2$ columns. 
        Specifically,
        \[
        \mathbf{Q}_{i + n/2} = \mathbf{Q}_i \quad \text{for $i \in [n/2]$}.
        \]
    
        The top half of $\mathbf{Q}_k$ corresponds to a column of $\mathbf{A}$, say $\mathbf{W}_k$.
        The bottom half corresponds to a column of $\mathbf{A}_{\rm opp}$, say $\mathbf{W}'_k$.
        By construction $\mathbf{W}_k + \mathbf{W}'_k = \mathbf{1}$ (vector of $n$ ones), because $\mathbf{A} + \mathbf{A}_{\rm opp} = \mathbf{J}$.
        Using the property that column $k$ and column $k+n/2$ are identical, we observe the following on the dot product $\mathbf{Q}_i \cdot \mathbf{Q}_j$:
        \[
        \mathbf{Q}_i \cdot \mathbf{Q}_j = 
        \begin{cases} 
        n/2 & \text{if } i = j \text{ or } i \equiv j \pmod{n/2}, \\
        n/4 & \text{otherwise (due to Hadamard orthogonality property)}.
        \end{cases}
        \]
    
        Thus, the expression simplifies to:
        \begin{align*}
            \|\mathbf{\tilde{A}}\mathbf{x}\|_2^2 &= \sum_{i=1}^{n/2} \sum_{j=1}^{n/2} x_i x_j (\mathbf{Q}_i \cdot \mathbf{Q}_j) + \sum_{i=1}^{n/2} \sum_{j=n/2+1}^{n} x_i x_j (\mathbf{Q}_i \cdot \mathbf{Q}_j) \\
            &\quad\quad + \sum_{i=n/2+1}^{n} \sum_{j=1}^{n/2} x_i x_j (\mathbf{Q}_i \cdot \mathbf{Q}_j) + \sum_{i=n/2+1}^{n} \sum_{j=n/2+1}^{n} x_i x_j (\mathbf{Q}_i \cdot \mathbf{Q}_j) \\
            &= \sum_{i=1}^{n/2} \sum_{j=1}^{n/2} x_i x_j \left( \frac{n}{4} \right) + \sum_{i=1}^{n/2} x_i^2 \left( \frac{n}{4} \right)\\
            &\quad \quad + \sum_{i=1}^{n/2} \sum_{j=n/2+1}^{n} x_i x_j \left( \frac{n}{4} \right) + \sum_{i=1}^{n/2} x_i x_{i+n/2} \left( \frac{n}{4} \right) \\
            &\quad\quad + \sum_{i=n/2+1}^{n} \sum_{j=1}^{n/2} x_i x_j \left( \frac{n}{4} \right) + \sum_{i=n/2+1}^{n} x_{i-n/2} x_i \left( \frac{n}{4} \right) \\
            &\quad \quad + \sum_{i=n/2+1}^{n} \sum_{j=n/2+1}^{n} x_i x_j \left( \frac{n}{4} \right) + \sum_{i=n/2+1}^{n} x_i^2 \left( \frac{n}{4} \right) \\
            &= \frac{n}{4} \left( \sum_{i=1}^n x_i \right)^2 + \frac{n}{4} \sum_{i=1}^{n/2} (x_i + x_{i+n/2})^2.
        \end{align*}
    \end{proof}
    
    Applying \Cref{clm:square-of-l2-norm}, we have:
    \begin{align*}
        \lVert \mathbf{B} (p \cdot \mathbf{1} - \mathbf{x}) \rVert_2^2 = \lVert \mathbf{\tilde{A}} (pt \cdot \mathbf{1} - \mathbf{z}) \rVert_2^2 
        &= \frac{n}{4} \left( \sum_{i=1}^n (pt - z_i) \right)^2 + \frac{n}{4} \sum_{i=1}^{n/2} \left( (pt - z_i) + (pt - z_{i+n/2}) \right)^2 \\
        &= \frac{n}{4} \left( \sum_{i=1}^n (pt - z_i) \right)^2 + \frac{n}{4} \sum_{i=1}^{n/2} \left( 0.5 - (z_i + z_{i+n/2}) \right)^2.
    \end{align*}
    Since $p=1/n$ and $t=n/4$, we have $pt = 1/4$.
    Thus, the vector $(pt \cdot \mathbf{1} - \mathbf{z})$ has every coordinate in the set $\{0.25, -0.75, -1.75, \dots\}$.
    Substituting this, we have
    \begin{align*}
        \lVert \mathbf{B} (p \cdot \mathbf{1} - \mathbf{x}) \rVert_2^2 &= \frac{n}{4} \left( \sum_{i=1}^n (pt - z_i) \right)^2 + \frac{n}{4} \sum_{i=1}^{n/2} \left( 0.5 - (z_i + z_{i+n/2}) \right)^2
        &\ge 0 + \frac{n}{4} \cdot \frac{n}{2} \cdot 0.5^2 = n \cdot \frac{n}{32}.
    \end{align*}
    
    Plugging the inequalities back together, we have
    \[
    \lVert \mathbf{B} (p \cdot \mathbf{1} - \mathbf{x}) \rVert_\infty \ge \frac{1}{\sqrt{n}} \lVert \mathbf{B} (p \cdot \mathbf{1} - \mathbf{x}) \rVert_2 \ge \frac{1}{\sqrt{n}} \sqrt{\frac{n^2}{32}} = \sqrt{\frac{n}{32}}.
    \]
\endgroup 
\end{proof}

\Cref{lem:weighted-disc-bound} implies that for any \part (represented by an arbitrary vector $\mathbf{x}$), there exists an agent whose valuation for that part (represented by the corresponding row entry) deviates significantly from the proportional share. 

Let $P_j$ be the smallest part in $P$.
Thus, $|P_j| \le m/n$. 
\Cref{lem:weighted-disc-bound} applied on $\mathbf{x}=\chi(P_j)$ implies that there exists some agent $a$ such that $|v_a(P_j) - \prop_a| \ge \sqrt{n/32}$. 
We now show that this large deviation implies the existence of an agent with \emph{low} value for $P_j$. We analyze the two cases:
\begin{itemize}
    \item \textbf{Case 1 (Negative Deviation):} If $v_a(P_j) \le \prop_a - \sqrt{n/32}$, the theorem holds immediately.
    
    \item \textbf{Case 2 (Positive Deviation):} Suppose that $v_a(P_j) > \prop_a + \sqrt{n/32}$.
    Consider the ``opposite'' agent $a_{\rm opp}$: if agent $a$ corresponds to row $r$ of $\mathbf{B}$, then $a_{\rm opp}$ corresponds to row $(r+n/2-1) \mod n + 1$. 
    Since $v_a(g) + v_{a_{\rm opp}}(g) = 1$ for all $g$, we have $v_{a_{\rm opp}}(P_j) = |P_j| - v_a(P_j)$.
    Substituting the bounds:
     $   v_{a_{\rm opp}}(P_j) = |P_j| - v_a(P_j) < \left( \frac{m}{n} \right) - \left(\prop_a + \sqrt{\frac{n}{32}}\right)$.
    
    Since $v_a(P_j) > \prop_a + \sqrt{n/32}$ and $P_j$ is of size at most $m/n$, we have that $a$ is not an all-ones or an all-zeros agent (since by \Cref{lem:prop-share-in-instance} such an agent has $\prop_a = m/n$ with $v_a(P_j) = |P_j| \le m/n$, or $\prop_a = 0$ with $v_a(P_j) = 0$, respectively). Moreover, \Cref{lem:prop-share-in-instance} also implies that $\prop_a = \prop_{a_{\rm opp}} = \frac{m}{2n}$.
    Substituting this back,
    \begin{align*}
        v_{a_{\rm opp}}(P_j) &< \frac{m}{n} - \left(\frac{m}{2n} + \sqrt{\frac{n}{32}}\right) 
        = \frac{m}{2n} - \sqrt{\frac{n}{32}} 
        = \prop_{a_{\rm opp}} - \sqrt{\frac{n}{32}} 
    \end{align*}
  Thus, we have $v_{a_{\rm opp}}(P_j) < \prop_{a_{\rm opp}} - \sqrt{n/32}$, and the theorem is proved.
\end{itemize}

\hide{This lemma implies that for any \part (represented by an arbitrary vector $x$), there exists an agent whose valuation for that part (represented by the corresponding row entry) has large deviation from its proportional share. We complete the proof by arguing that even if this deviation is positive, the ``opposite'' agent experiences a negative deviation.}


\ifthenelse{\equal{\version}{uncompressed}}{
\begin{proof}We note that $\|\mathbf{v}\|_\infty \ge \frac{1}{\sqrt{n}} \|\mathbf{v}\|_2$, for any $v\in \{0,1\}^{nt}$. Thus, we need to lower bound \hide{ the $\ell_2$ norm of} $\lVert\mathbf{B}(1/n \cdot \mathbf{1} - \mathbf{x})\rVert_2$.
Recall that $t = n/4$ and $\mathbf{B} \coloneq [\mathbf{\tilde{A}} ~|~ \cdots ~|~ \mathbf{\tilde{A}}]$.
    Consider any $\mathbf{x} \in \{0,1\}^{nt}$ and let $p \coloneq 1/n$.
    Let $\mathbf{z} \in \{0,\dots,t\}^n$ be defined by $z_i = \sum_{j=0}^{t-1} x_{i+nj}$.
    Thus, $\mathbf{B}(p \cdot \mathbf{1} - \mathbf{x}) = \mathbf{\tilde{A}}(pt \cdot \mathbf{1} - \mathbf{z})$.
   
    In order to lower bound $\lVert\mathbf{B}(1/n \cdot \mathbf{1} - \mathbf{x})\rVert_2$, we will analyze the $\ell_2$-norm of matrix-vector products involving the block matrix $\mathbf{\tilde{A}}$.
    
    \begin{restatable}[\app]{claim-inside-lemma}{squareofltwonorm}
        \label{clm:square-of-l2-norm}
        For any vector $\mathbf{y} \in \mathbb{R}^n$,
        \[
        \lVert \mathbf{\tilde{A}}\mathbf{y} \rVert_2^2 = \frac{n}{4} \left( \sum_{i=1}^n y_i \right)^2 + \frac{n}{4} \sum_{i=1}^{n/2} \left( y_i + y_{i+n/2} \right)^2.
        \]
    \end{restatable}
    \ifthenelse{\equal{\version}{uncompressed}}{
        \begin{proof}
        Expanding the LHS, we have:
        \begin{align*}
            \|\mathbf{\tilde{A}}\mathbf{x}\|_2^2 &= (\mathbf{\tilde{A}}\mathbf{x})^\top (\mathbf{\tilde{A}}\mathbf{x}) \\
            &= \mathbf{x}^\top \mathbf{\tilde{A}}^\top \mathbf{\tilde{A}} \mathbf{x} \\
            &= \sum_{i=1}^n \sum_{j=1}^n x_i x_j \left( \mathbf{\tilde{A}}_{(\cdot, i)} \cdot \mathbf{\tilde{A}}_{(\cdot, j)} \right),
        \end{align*}
        where $\mathbf{\tilde{A}}_{(\cdot, k)}$ denotes the $k$-th column of $\mathbf{\tilde{A}}$.
    
        Let $\mathbf{Q}_k$ denote the $k$-th column of the matrix $\mathbf{\tilde{A}}$. 
        Expanding the double sum by splitting indices into the left half ($1, \dots, n/2$) and right half ($n/2+1, \dots, n$), we have:
        \begin{align*}
            \|\mathbf{\tilde{A}}\mathbf{x}\|_2^2 &= \sum_{i=1}^{n/2} \sum_{j=1}^{n/2} x_i x_j (\mathbf{Q}_i \cdot \mathbf{Q}_j) + \sum_{i=1}^{n/2} \sum_{j=n/2+1}^{n} x_i x_j (\mathbf{Q}_i \cdot \mathbf{Q}_j) \\
            &\quad\quad + \sum_{i=n/2+1}^{n} \sum_{j=1}^{n/2} x_i x_j (\mathbf{Q}_i \cdot \mathbf{Q}_j) + \sum_{i=n/2+1}^{n} \sum_{j=n/2+1}^{n} x_i x_j (\mathbf{Q}_i \cdot \mathbf{Q}_j).
        \end{align*}
    
        Note that the structure of $\mathbf{\tilde{A}}$ implies that the first $n/2$ columns are identical to the second $n/2$ columns. 
        Specifically,
        \[
        \mathbf{Q}_{i + n/2} = \mathbf{Q}_i \quad \text{for $i \in [n/2]$}.
        \]
    
        The top half of $\mathbf{Q}_k$ corresponds to a column of $\mathbf{A}$, say $\mathbf{W}_k$.
        The bottom half corresponds to a column of $\mathbf{A}_{\rm opp}$, say $\mathbf{W}'_k$.
        By construction $\mathbf{W}_k + \mathbf{W}'_k = \mathbf{1}$ (vector of $n$ ones), because $\mathbf{A} + \mathbf{A}_{\rm opp} = \mathbf{J}$.
        Using the property that column $k$ and column $k+n/2$ are identical, we observe the following on the dot product $\mathbf{Q}_i \cdot \mathbf{Q}_j$:
        \[
        \mathbf{Q}_i \cdot \mathbf{Q}_j = 
        \begin{cases} 
        n/2 & \text{if } i = j \text{ or } i \equiv j \pmod{n/2}, \\
        n/4 & \text{otherwise (due to Hadamard orthogonality property)}.
        \end{cases}
        \]
    
        Thus, the expression simplifies to:
        \begin{align*}
            \|\mathbf{\tilde{A}}\mathbf{x}\|_2^2 &= \sum_{i=1}^{n/2} \sum_{j=1}^{n/2} x_i x_j (\mathbf{Q}_i \cdot \mathbf{Q}_j) + \sum_{i=1}^{n/2} \sum_{j=n/2+1}^{n} x_i x_j (\mathbf{Q}_i \cdot \mathbf{Q}_j) \\
            &\quad\quad + \sum_{i=n/2+1}^{n} \sum_{j=1}^{n/2} x_i x_j (\mathbf{Q}_i \cdot \mathbf{Q}_j) + \sum_{i=n/2+1}^{n} \sum_{j=n/2+1}^{n} x_i x_j (\mathbf{Q}_i \cdot \mathbf{Q}_j) \\
            &= \sum_{i=1}^{n/2} \sum_{j=1}^{n/2} x_i x_j \left( \frac{n}{4} \right) + \sum_{i=1}^{n/2} x_i^2 \left( \frac{n}{4} \right)\\
            &\quad \quad + \sum_{i=1}^{n/2} \sum_{j=n/2+1}^{n} x_i x_j \left( \frac{n}{4} \right) + \sum_{i=1}^{n/2} x_i x_{i+n/2} \left( \frac{n}{4} \right) \\
            &\quad\quad + \sum_{i=n/2+1}^{n} \sum_{j=1}^{n/2} x_i x_j \left( \frac{n}{4} \right) + \sum_{i=n/2+1}^{n} x_{i-n/2} x_i \left( \frac{n}{4} \right) \\
            &\quad \quad + \sum_{i=n/2+1}^{n} \sum_{j=n/2+1}^{n} x_i x_j \left( \frac{n}{4} \right) + \sum_{i=n/2+1}^{n} x_i^2 \left( \frac{n}{4} \right) \\
            &= \frac{n}{4} \left( \sum_{i=1}^n x_i \right)^2 + \frac{n}{4} \sum_{i=1}^{n/2} (x_i + x_{i+n/2})^2.
        \end{align*}
    \end{proof}
    }{}
    
    Applying \Cref{clm:square-of-l2-norm}, we have:
    \begin{align*}
        \lVert \mathbf{B} (p \cdot \mathbf{1} - \mathbf{x}) \rVert_2^2 = \lVert \mathbf{\tilde{A}} (pt \cdot \mathbf{1} - \mathbf{z}) \rVert_2^2 
        &= \frac{n}{4} \left( \sum_{i=1}^n (pt - z_i) \right)^2 + \frac{n}{4} \sum_{i=1}^{n/2} \left( (pt - z_i) + (pt - z_{i+n/2}) \right)^2 \\
        &= \frac{n}{4} \left( \sum_{i=1}^n (pt - z_i) \right)^2 + \frac{n}{4} \sum_{i=1}^{n/2} \left( 0.5 - (z_i + z_{i+n/2}) \right)^2.
    \end{align*}
    Since $p=1/n$ and $t=n/4$, we have $pt = 1/4$.
    Thus, the vector $(pt \cdot \mathbf{1} - \mathbf{z})$ has every coordinate in the set $\{0.25, -0.75, -1.75, \dots\}$.
    Substituting this, we have
    \begin{align*}
        \lVert \mathbf{B} (p \cdot \mathbf{1} - \mathbf{x}) \rVert_2^2 &= \frac{n}{4} \left( \sum_{i=1}^n (pt - z_i) \right)^2 + \frac{n}{4} \sum_{i=1}^{n/2} \left( 0.5 - (z_i + z_{i+n/2}) \right)^2
        &\ge 0 + \frac{n}{4} \cdot \frac{n}{2} \cdot 0.5^2 = n \cdot \frac{n}{32}.
    \end{align*}
    
    Plugging the inequalities back together, we have
    \[
    \lVert \mathbf{B} (p \cdot \mathbf{1} - \mathbf{x}) \rVert_\infty \ge \frac{1}{\sqrt{n}} \lVert \mathbf{B} (p \cdot \mathbf{1} - \mathbf{x}) \rVert_2 \ge \frac{1}{\sqrt{n}} \sqrt{\frac{n^2}{32}} = \sqrt{\frac{n}{32}}.
    \]
\end{proof}

This lemma implies that for any \part (represented by an arbitrary vector $x$), there exists an agent whose valuation for that part (represented by the corresponding row entry) deviates significantly from the proportional share. 

\org{suppress from here}
Let $P_j$ be the smallest part in $P$.
Thus, $|P_j| \le m/n$. 
\Cref{lem:weighted-disc-bound} applied on $\mathbf{x}=\chi(P_j)$ implies that there exists some agent $a$ such that $|v_a(P_j) - \prop_a| \ge \sqrt{n/32}$. 
We now show that this large deviation implies the existence of an agent with \emph{low} value for $P_j$. We analyze the two cases:
\begin{itemize}
    \item \textbf{Case 1 (Negative Deviation):} If $v_a(P_j) \le \prop_a - \sqrt{n/32}$, the theorem holds immediately.
    
    \item \textbf{Case 2 (Positive Deviation):} Suppose that $v_a(P_j) > \prop_a + \sqrt{n/32}$.
    Consider the ``opposite'' agent $a_{\rm opp}$: if agent $a$ corresponds to row $r$ of $\mathbf{B}$, then $a_{\rm opp}$ corresponds to row $(r+n/2-1) \mod n + 1$. 
    Since $v_a(g) + v_{a_{\rm opp}}(g) = 1$ for all $g$, we have $v_{a_{\rm opp}}(P_j) = |P_j| - v_a(P_j)$.
    Substituting the bounds:
     $   v_{a_{\rm opp}}(P_j) = |P_j| - v_a(P_j) < \left( \frac{m}{n} \right) - \left(\prop_a + \sqrt{\frac{n}{32}}\right)$.
    
    Since $v_a(P_j) > \prop_a + \sqrt{n/32}$ and $P_j$ is of size at most $m/n$, we have that $a$ is not an all-ones or an all-zeros agent (since by \Cref{lem:prop-share-in-instance} such an agent has $\prop_a = m/n$ with $v_a(P_j) = |P_j| \le m/n$, or $\prop_a = 0$ with $v_a(P_j) = 0$, respectively). Moreover, \Cref{lem:prop-share-in-instance} also implies that $\prop_a = \prop_{a_{\rm opp}} = \frac{m}{2n}$.
    Substituting this back,
    \begin{align*}
        v_{a_{\rm opp}}(P_j) &< \frac{m}{n} - \left(\frac{m}{2n} + \sqrt{\frac{n}{32}}\right) 
        = \frac{m}{2n} - \sqrt{\frac{n}{32}} 
        = \prop_{a_{\rm opp}} - \sqrt{\frac{n}{32}} 
    \end{align*}
  Thus, we have $v_{a_{\rm opp}}(P_j) < \prop_{a_{\rm opp}} - \sqrt{n/32}$, and the theorem is proved.
\end{itemize}
}{}
\section{Directions for future work}

In this article, we laid the groundwork for the model where agents pick bundles autonomously; and there are multiple directions to pursue, as noted below. 

\begin{description}[wide=0pt]
    \item [{\it Generalization}] beyond additive non-negative valuations to more general valuations such as submodular or subadditive valuations, as well as extending the study to fair division of chores and mixed items. 

    \item [{\it Exploring randomization,}] both for creating partitions as well as for tie-breaking when agents pick a \part among those of highest value, to obtain higher guarantees.

    \item[{\it  Analyzing ``tail guarantees'',}] i.e, obtaining higher guarantees for all but a few agents that come last.

    \item [{\it Impossibilities or improved guarantees:}]For the case where our proven guarantee is $\prop_i - \omega(1)$ (e.g., relaxations of ordered additive or when proportional share is bounded), the question remains whether there exists an algorithm that guarantees each agent $i$ a part of value at least $\prop_i - \bigoh(1)$. 
    Note that our analysis of our algorithms is tight, as discussed at the end of \Cref{subsec:contribution}.
    However, an entirely different algorithmic approach, perhaps one that goes beyond polynomial-time, may yield better guarantees; and alternatively, an impossibility result will rule that out. 

   \item [{\it Strategic behavior}]\hide{We also note the importance of addressing strategic behavior} of agents--unilateral or in collusion with others--where agents misreport preferences to obtain higher value \parts for themselves.

    \item [{\it Optimization variant:}] Deciding if there exists a dynamic partition (when the arrival order is fixed) that guarantees each agent a part of value at least their proportional share is \textsf{NP}-hard, even when all agent preferences are identical. 
    This follows from scaling the standard reduction from the Partition problem for the classical setting where agents are assigned bundles of items~\cite{demko1988equitable}. 
    Consequently, it is natural to consider the optimization variant where the objective is to maximize the surplus (or minimize the deficit) that each agent is guaranteed relative to their proportional share. Formally, given a fair division instance with a fixed arrival order, we seek the largest integer $t$ (possibly negative) such that there exists a partition where each agent $i\in [n]$ receives a part of value at least $\prop_i + t$. The aforementioned reduction strategy proves that this problem is \textsf{NP}-hard even when $t=0$, $n=2$, and the number of agent types is one.
    A worthy goal would be to explore efficient algorithms, perhaps parameterized exact algorithms or even parameterized approximation algorithms, with respect to the structural parameters we have identified. 
    These include the number of agent/item types, swap distance, the depth of the laminar swap family, or the number of linearly separable pieces for the ordered additive relaxations; the degree of influence of agents, or the larger parameter maximum degree, and the maximum size of the edges, to name a few for the hypergraph setting; and the maximum {\it influence of an item} defined as the number of agents that value that item positively for the case when all agents have bounded indifference. 
\end{description}

\bibliographystyle{alpha}
\bibliography{bibliography}

@inproceedings{10.5555/3709347.3743514,
author = {Afshinmehr, Mahyar and Danaei, Alireza and Kazemi, Mehrafarin and Mehlhorn, Kurt and Rathi, Nidhi},
title = {EFX Allocations and Orientations on Bipartite Multi-graphs: A Complete Picture},
year = {2025},
isbn = {9798400714269},
publisher = {International Foundation for Autonomous Agents and Multiagent Systems},
address = {Richland, SC},
booktitle = {Proceedings of the 24th International Conference on Autonomous Agents and Multiagent Systems},
pages = {32–40},
numpages = {9},
location = {Detroit, MI, USA},
series = {AAMAS '25}
}

@article{DBLP:journals/ai/AmanatidisABFLMVW23,
  author       = {Georgios Amanatidis and
                  Haris Aziz and
                  Georgios Birmpas and
                  Aris Filos{-}Ratsikas and
                  Bo Li and
                  Herv{\'{e}} Moulin and
                  Alexandros A. Voudouris and
                  Xiaowei Wu},
  title        = {Fair division of indivisible goods: Recent progress and open questions},
  journal      = {Artif. Intell.},
  volume       = {322},
  pages        = {103965},
  year         = {2023},
  url          = {https://doi.org/10.1016/j.artint.2023.103965},
  doi          = {10.1016/J.ARTINT.2023.103965},
  bibsource    = {dblp computer science bibliography, https://dblp.org}
}

@article{DBLP:journals/corr/abs-2307-10985,
  author       = {Shaily Mishra and
                  Manisha Padala and
                  Sujit Gujar},
  title        = {Fair Allocation of goods and chores - Tutorial and Survey of Recent
                  Results},
  journal      = {CoRR},
  volume       = {abs/2307.10985},
  year         = {2023},
  url          = {https://doi.org/10.48550/arXiv.2307.10985},
  doi          = {10.48550/ARXIV.2307.10985},
  eprinttype    = {arXiv},
  eprint       = {2307.10985},
  bibsource    = {dblp computer science bibliography, https://dblp.org}
}

@inproceedings{DBLP:conf/sigecom/BabaioffM25,
  author       = {Moshe Babaioff and
                  Noam Manaker Morag},
  editor       = {Itai Ashlagi and
                  Aaron Roth},
  title        = {On Truthful Mechanisms without Pareto-efficiency: Characterizations
                  and Fairness},
  booktitle    = {Proceedings of the 26th {ACM} Conference on Economics and Computation,
                  {EC} 2025, Stanford University, Stanford, CA, USA, July 7-10, 2025},
  pages        = {447},
  publisher    = {{ACM}},
  year         = {2025},
  url          = {https://doi.org/10.1145/3736252.3742569},
  doi          = {10.1145/3736252.3742569},
  bibsource    = {dblp computer science bibliography, https://dblp.org}
}

@inproceedings{DBLP:conf/aaai/BaklanovGGS21,
  author       = {Artem Baklanov and
                  Pranav Garimidi and
                  Vasilis Gkatzelis and
                  Daniel Schoepflin},
  title        = {Achieving Proportionality up to the Maximin Item with Indivisible
                  Goods},
  booktitle    = {Thirty-Fifth {AAAI} Conference on Artificial Intelligence, {AAAI}
                  2021, Thirty-Third Conference on Innovative Applications of Artificial
                  Intelligence, {IAAI} 2021, The Eleventh Symposium on Educational Advances
                  in Artificial Intelligence, {EAAI} 2021},
  pages        = {5143--5150},
  publisher    = {{AAAI} Press},
  year         = {2021},
  url          = {https://doi.org/10.1609/aaai.v35i6.16650},
  doi          = {10.1609/AAAI.V35I6.16650},
  bibsource    = {dblp computer science bibliography, https://dblp.org}
}

@inproceedings{bhaskar2024efxallocationsmultigraphclasses,
  author       = {Umang Bhaskar and
                  Yeshwant Pandit},
  editor       = {C. Aiswarya and
                  Ruta Mehta and
                  Subhajit Roy},
  title        = {Extending {EFX} Allocations to Further Multi-Graph Classes},
  booktitle    = {45th {IARCS} Annual Conference on Foundations of Software Technology
                  and Theoretical Computer Science, {FSTTCS} 2025},
  series       = {LIPIcs},
  volume       = {360},
  pages        = {15:1--15:18},
  publisher    = {Schloss Dagstuhl - Leibniz-Zentrum f{\"{u}}r Informatik},
  year         = {2025},
  url          = {https://doi.org/10.4230/LIPIcs.FSTTCS.2025.15},
  doi          = {10.4230/LIPICS.FSTTCS.2025.15},
  bibsource    = {dblp computer science bibliography, https://dblp.org}
}

@inproceedings{DBLP:conf/sagt/BismuthBS24,
  author       = {Samuel Bismuth and
                  Ivan Bliznets and
                  Erel Segal{-}Halevi},
  editor       = {Guido Sch{\"{a}}fer and
                  Carmine Ventre},
  title        = {Fair Division with Bounded Sharing: Binary and Non-degenerate Valuations},
  booktitle    = {Algorithmic Game Theory - 17th International Symposium, {SAGT} 2024},
  series       = {Lecture Notes in Computer Science},
  volume       = {15156},
  pages        = {89--107},
  publisher    = {Springer},
  year         = {2024},
  url          = {https://doi.org/10.1007/978-3-031-71033-9\_6},
  doi          = {10.1007/978-3-031-71033-9\_6},
  bibsource    = {dblp computer science bibliography, https://dblp.org}
}

@book{DBLP:books/daglib/0017730,
  author       = {Steven J. Brams and
                  Alan D. Taylor},
  title        = {Fair division - from cake-cutting to dispute resolution},
  publisher    = {Cambridge University Press},
  year         = {1996},
  isbn         = {978-0-521-55644-6},
  bibsource    = {dblp computer science bibliography, https://dblp.org}
}

@inproceedings{DBLP:conf/sigecom/0001FKS23,
  author       = {George Christodoulou and
                  Amos Fiat and
                  Elias Koutsoupias and
                  Alkmini Sgouritsa},
  
  title        = {Fair allocation in graphs},
  booktitle    = {Proceedings of the 24th {ACM} Conference on Economics and Computation,
                  {EC} },
  pages        = {473--488},
  publisher    = {{ACM}},
  year         = {2023},
  url          = {https://doi.org/10.1145/3580507.3597764},
  doi          = {10.1145/3580507.3597764},
  bibsource    = {dblp computer science bibliography, https://dblp.org}
}

@article{DBLP:journals/teco/CaragiannisKMPS19,
  author       = {Ioannis Caragiannis and
                  David Kurokawa and
                  Herv{\'{e}} Moulin and
                  Ariel D. Procaccia and
                  Nisarg Shah and
                  Junxing Wang},
  title        = {The Unreasonable Fairness of Maximum Nash Welfare},
  journal      = {{ACM} Trans. Economics and Comput.},
  volume       = {7},
  number       = {3},
  pages        = {12:1--12:32},
  year         = {2019},
  url          = {https://doi.org/10.1145/3355902},
  doi          = {10.1145/3355902},
  bibsource    = {dblp computer science bibliography, https://dblp.org}
}

@inproceedings{DBLP:conf/sigecom/ConitzerF017,
  author       = {Vincent Conitzer and
                  Rupert Freeman and
                  Nisarg Shah},
  editor       = {Constantinos Daskalakis and
                  Moshe Babaioff and
                  Herv{\'{e}} Moulin},
  title        = {Fair Public Decision Making},
  booktitle    = {Proceedings of the 2017 {ACM} Conference on Economics and Computation,
                  {EC} '17, 2017},
  pages        = {629--646},
  publisher    = {{ACM}},
  year         = {2017},
  url          = {https://doi.org/10.1145/3033274.3085125},
  doi          = {10.1145/3033274.3085125},
  bibsource    = {dblp computer science bibliography, https://dblp.org}
}

@article{DBLP:journals/almob/CunhaSS24,
  author       = {Lu{\'{\i}}s Cunha and
                  Ignasi Sau and
                  U{\'{e}}verton S. Souza},
  title        = {On the parameterized complexity of the median and closest problems
                  under some permutation metrics},
  journal      = {Algorithms Mol. Biol.},
  volume       = {19},
  number       = {1},
  pages        = {24},
  year         = {2024},
  url          = {https://doi.org/10.1186/s13015-024-00269-z},
  doi          = {10.1186/S13015-024-00269-Z},
  bibsource    = {dblp computer science bibliography, https://dblp.org}
}

@article{demko1988equitable,
  title={Equitable distribution of indivisible objects},
  author={Demko, Stephen and Hill, Theodore P},
  journal={Mathematical Social Sciences},
  volume={16},
  number={2},
  pages={145--158},
  year={1988},
  publisher={Elsevier}
}

@article{DBLP:journals/cpc/DoerrS03,
  author       = {Benjamin Doerr and
                  Anand Srivastav},
  title        = {Multicolour Discrepancies},
  journal      = {Comb. Probab. Comput.},
  volume       = {12},
  number       = {4},
  pages        = {365--399},
  year         = {2003},
  url          = {https://doi.org/10.1017/S0963548303005662},
  doi          = {10.1017/S0963548303005662},
  bibsource    = {dblp computer science bibliography, https://dblp.org}
}

@article{hedayat1978hadamard,
  title={Hadamard matrices and their applications},
  author={Hedayat, A and Wallis, Walter Dennis},
  journal={The annals of statistics},
  pages={1184--1238},
  year={1978},
  publisher={JSTOR}
}

@misc{hsu2025efxorientationsmultigraphs,
      title={EFX Orientations of Multigraphs}, 
      author={Kevin Hsu},
      year={2025},
      eprint={2410.12039},
      archivePrefix={arXiv},
      primaryClass={cs.GT},
      url={https://arxiv.org/abs/2410.12039}, 
}

@article{DBLP:journals/tcs/Jerrum85,
  author       = {Mark Jerrum},
  title        = {The Complexity of Finding Minimum-Length Generator Sequences},
  journal      = {Theor. Comput. Sci.},
  volume       = {36},
  pages        = {265--289},
  year         = {1985},
  url          = {https://doi.org/10.1016/0304-3975(85)90047-7},
  doi          = {10.1016/0304-3975(85)90047-7},
  bibsource    = {dblp computer science bibliography, https://dblp.org}
}

@article{DBLP:journals/tcs/ManurangsiS22,
  author       = {Pasin Manurangsi and
                  Warut Suksompong},
  title        = {Almost envy-freeness for groups: Improved bounds via discrepancy theory},
  journal      = {Theor. Comput. Sci.},
  volume       = {930},
  pages        = {179--195},
  year         = {2022},
  url          = {https://doi.org/10.1016/j.tcs.2022.07.022},
  doi          = {10.1016/J.TCS.2022.07.022},
  bibsource    = {dblp computer science bibliography, https://dblp.org}
}

@inproceedings{doi:10.1137/1.9781611978964.20Manurangsi26,
author = {Pasin Manurangsi and Raghu Meka},
title = {Tight Lower Bound for Multicolor Discrepancy},
booktitle = {2026 SIAM Symposium on Simplicity in Algorithms (SOSA)},
chapter = {},
year={2026},
pages = {266-274},
doi = {10.1137/1.9781611978964.20},
URL = {https://epubs.siam.org/doi/abs/10.1137/1.9781611978964.20},
eprint = {https://epubs.siam.org/doi/pdf/10.1137/1.9781611978964.20},
}

@article{DBLP:journals/tcs/Popov07,
  author       = {V. Y. Popov},
  title        = {Multiple genome rearrangement by swaps and by element duplications},
  journal      = {Theor. Comput. Sci.},
  volume       = {385},
  number       = {1-3},
  pages        = {115--126},
  year         = {2007},
  url          = {https://doi.org/10.1016/j.tcs.2007.05.029},
  doi          = {10.1016/J.TCS.2007.05.029},
  bibsource    = {dblp computer science bibliography, https://dblp.org}
}

@article{DBLP:journals/tcs/Schlotter24,
  author       = {Ildik{\'{o}} Schlotter},
  title        = {Recognizing when a preference system is close to admitting a master
                  list},
  journal      = {Theor. Comput. Sci.},
  volume       = {994},
  pages        = {114445},
  year         = {2024},
  url          = {https://doi.org/10.1016/j.tcs.2024.114445},
  doi          = {10.1016/J.TCS.2024.114445},
  bibsource    = {dblp computer science bibliography, https://dblp.org}
}

@article{DBLP:journals/aamas/BouveretL16,
  author       = {Sylvain Bouveret and
                  Michel Lema{\^{\i}}tre},
  title        = {Characterizing conflicts in fair division of indivisible goods using
                  a scale of criteria},
  journal      = {Auton. Agents Multi Agent Syst.},
  volume       = {30},
  number       = {2},
  pages        = {259--290},
  year         = {2016},
  url          = {https://doi.org/10.1007/s10458-015-9287-3},
  doi          = {10.1007/S10458-015-9287-3},
  bibsource    = {dblp computer science bibliography, https://dblp.org}
}

@article{DBLP:journals/corr/abs-2506-20317,
  author       = {George Christodoulou and
                  Symeon Mastrakoulis},
  title        = {Exact and approximate maximin share allocations in multi-graphs},
  journal      = {CoRR},
  volume       = {abs/2506.20317},
  year         = {2025},
  url          = {https://doi.org/10.48550/arXiv.2506.20317},
  doi          = {10.48550/ARXIV.2506.20317},
  eprinttype    = {arXiv},
  eprint       = {2506.20317},
  bibsource    = {dblp computer science bibliography, https://dblp.org}
}

@inproceedings{DBLP:conf/ijcai/0027WL024,
  author       = {Yu Zhou and
                  Tianze Wei and
                  Minming Li and
                  Bo Li},
  title        = {A Complete Landscape of {EFX} Allocations on Graphs: Goods, Chores
                  and Mixed Manna},
  booktitle    = {Proceedings of the Thirty-Third International Joint Conference on
                  Artificial Intelligence, {IJCAI} 2024},
  pages        = {3049--3056},
  publisher    = {ijcai.org},
  year         = {2024},
  url          = {https://www.ijcai.org/proceedings/2024/338},
  bibsource    = {dblp computer science bibliography, https://dblp.org}
}

@inproceedings{DBLP:conf/ifaamas/SgouritsaS25,
  author       = {Alkmini Sgouritsa and
                  Minas Marios Sotiriou},
  editor       = {Sanmay Das and
                  Ann Now{\'{e}} and
                  Yevgeniy Vorobeychik},
  title        = {On the Existence of {EFX} Allocations in Multigraphs},
  booktitle    = {Proceedings of the 24th International Conference on Autonomous Agents
                  and Multiagent Systems, {AAMAS} 2025},
  pages        = {2735--2737},
  publisher    = {International Foundation for Autonomous Agents and Multiagent Systems
                  / {ACM}},
  year         = {2025},
  url          = {https://dl.acm.org/doi/10.5555/3709347.3743995},
  doi          = {10.5555/3709347.3743995},
  bibsource    = {dblp computer science bibliography, https://dblp.org}
}

@inproceedings{DBLP:conf/ijcai/DeligkasEGK25,
  author       = {Argyrios Deligkas and
                  Eduard Eiben and
                  Tiger{-}Lily Goldsmith and
                  Viktoriia Korchemna},
  title        = {{EF1} and {EFX} Orientations},
  booktitle    = {Proceedings of the Thirty-Fourth International Joint Conference on
                  Artificial Intelligence, {IJCAI} 2025,
                  2025},
  pages        = {56--63},
  publisher    = {ijcai.org},
  year         = {2025},
  url          = {https://doi.org/10.24963/ijcai.2025/7},
  doi          = {10.24963/IJCAI.2025/7},
  bibsource    = {dblp computer science bibliography, https://dblp.org}
}

@article{DBLP:journals/corr/abs-2404-13527,
  author       = {Jinghan A. Zeng and
                  Ruta Mehta},
  title        = {On the structure of envy-free orientations on graphs},
  journal      = {CoRR},
  volume       = {abs/2404.13527},
  year         = {2024},
  url          = {https://doi.org/10.48550/arXiv.2404.13527},
  doi          = {10.48550/ARXIV.2404.13527},
  eprinttype    = {arXiv},
  eprint       = {2404.13527},
  bibsource    = {dblp computer science bibliography, https://dblp.org}
}

@inproceedings{DBLP:conf/aldt/MisraS24,
  author       = {Neeldhara Misra and
                  Aditi Sethia},
  editor       = {Rupert Freeman and
                  Nicholas Mattei},
  title        = {Envy-Free and Efficient Allocations for Graphical Valuations},
  booktitle    = {Algorithmic Decision Theory - 8th International Conference, {ADT}
                  2024},
  series       = {Lecture Notes in Computer Science},
  volume       = {15248},
  pages        = {258--272},
  publisher    = {Springer},
  year         = {2024},
  url          = {https://doi.org/10.1007/978-3-031-73903-3\_17},
  doi          = {10.1007/978-3-031-73903-3\_17},
  bibsource    = {dblp computer science bibliography, https://dblp.org}
}

@article{budish2011combinatorial,
  title={The combinatorial assignment problem: Approximate competitive equilibrium from equal incomes},
  author={Budish, Eric},
  journal={Journal of Political Economy},
  volume={119},
  number={6},
  pages={1061--1103},
  year={2011},
  publisher={University of Chicago Press Chicago, IL}
}

@article{DBLP:journals/jacm/KurokawaPW18,
  author       = {David Kurokawa and
                  Ariel D. Procaccia and
                  Junxing Wang},
  title        = {Fair Enough: Guaranteeing Approximate Maximin Shares},
  journal      = {J. {ACM}},
  volume       = {65},
  number       = {2},
  pages        = {8:1--8:27},
  year         = {2018},
  url          = {https://doi.org/10.1145/3140756},
  doi          = {10.1145/3140756},
  bibsource    = {dblp computer science bibliography, https://dblp.org}
}

@inproceedings{DBLP:conf/webi/PadalaG21,
  author       = {Manisha Padala and
                  Sujit Gujar},
  editor       = {Jing He and
                  Rainer Unland and
                  Eugene Santos Jr. and
                  Xiaohui Tao and
                  Hemant Purohit and
                  Willem{-}Jan van den Heuvel and
                  John Yearwood and
                  Jie Cao},
  title        = {Mechanism Design without Money for Fair Allocations},
  booktitle    = {{WI-IAT} '21: {IEEE/WIC/ACM} International Conference on Web Intelligence, 2021},
  pages        = {382--389},
  publisher    = {{ACM}},
  year         = {2021},
  url          = {https://doi.org/10.1145/3486622.3493955},
  doi          = {10.1145/3486622.3493955},
  bibsource    = {dblp computer science bibliography, https://dblp.org}
}

@inproceedings{DBLP:conf/nips/0001V22,
  author       = {Alexandros Psomas and
                  Paritosh Verma},
  editor       = {Sanmi Koyejo and
                  S. Mohamed and
                  A. Agarwal and
                  Danielle Belgrave and
                  K. Cho and
                  A. Oh},
  title        = {Fair and Efficient Allocations Without Obvious Manipulations},
  booktitle    = {Advances in Neural Information Processing Systems 35: Annual Conference
                  on Neural Information Processing Systems 2022, NeurIPS 2022},
  year         = {2022},
  url          = {http://papers.nips.cc/paper\_files/paper/2022/hash/57250222014c35949476f3f272c322d2-Abstract-Conference.html},
  bibsource    = {dblp computer science bibliography, https://dblp.org}
}

\appendix

\newpage

\section{Our baseline: an application of Multicolor Discrepancy}\label{sec:prelim-discrepancy}

Consider a matrix $\mathbf{A} \in [0,1]^{n \times m}$.
A $k$-coloring of the columns $[m]$ is a function $\chi: [m] \rightarrow [k]$.
Let $\mathbf{1}$ denote the $m$-length all-ones vector and $\mathbf{1}(S)$ denote the indicator vector of a set $S \subseteq [m]$.

\begin{definition}[$k$-color discrepancy]
\[
\disc(\mathbf{A},k) 
\coloneqq 
\min_{\chi: [m] \rightarrow [k]} \max_{s \in [k]}
\left\lVert
\mathbf{A}\left(\frac1k \cdot \mathbf{1} - \mathbf{1}(\chi^{-1}(s))\right)
\right\rVert_\infty
\]
\end{definition}

For the work in this paper, we use $n$-color discrepancy, where $n$ denotes the number of agents.

Manurangsi and Suksompong~\cite{DBLP:journals/tcs/ManurangsiS22} present a polynomial-time algorithm that computes a coloring with discrepancy $\bigoh(\sqrt{n})$. 
They note that this result was already implicitly known via existing literature.

\begin{proposition}[Theorem A.1 in \cite{DBLP:journals/tcs/ManurangsiS22}]
    \label{prop:manurangsiS-disc} 
    Given any $\mathbf{A} \in [0,1]^{n \times m}$ and $k \in \mathbb{N}$, there exists a deterministic polynomial-time algorithm that computes a coloring $\chi: [m] \rightarrow [k]$ such that 
    \[
    \left\lVert
    \mathbf{A}\left(\frac1k \cdot \mathbf{1} - \mathbf{1}(\chi^{-1}(s))\right)
    \right\rVert_\infty
    \le \bigoh(\sqrt{n}).
    \]
\end{proposition}

Applying the above to the valuation matrix $\mathbf{A} \in [0,1]^{n \times m}$ (where rows correspond to agents and columns to items) with $k=n$ colors, we obtain a coloring $\chi:[m]\rightarrow[n]$.
Consider the partition $M=(\M{1}, \dots, \M{n})$ where the part $\M{p}$ is the set of items colored $p$ in $\chi$.
Observe that the $a$-th component of the vector $\mathbf{A} (\frac{1}{n} \cdot \mathbf{1})$ is exactly $\frac{1}{n} \sum_{j=1}^m v_a(j) = \prop_a$. 
Similarly, the $a$-th component of $\mathbf{A} \mathbf{1}(\chi^{-1}(p))$ is exactly $v_a(\M{p})$.
Thus, the bound from \Cref{prop:manurangsiS-disc} implies that for each agent $a$ and part $\M{p}$:
\[
    \lvert \prop_a - v_a(\M{p})\rvert 
    =
    \left\lvert \left( \mathbf{A}\left(\frac1n \cdot \mathbf{1} - \mathbf{1}(\chi^{-1}(p))\right) \right)_a \right\rvert
    \le
    \left\lVert
    \mathbf{A}\left(\frac1n \cdot \mathbf{1} - \mathbf{1}(\chi^{-1}(p))\right)
    \right\rVert_\infty
    \le \bigoh(\sqrt{n}).
\]
This yields the following baseline guarantee for our setting.

\begin{proposition}
\label{prop:disc-upperbound}
There is a partition $M=(\M{1}, \dots, \M{n})$ such that
    $\prop_a-\bigoh(\sqrt{n}) \le v_a(\M{p}) \le \prop_a+\bigoh(\sqrt{n})$ for each agent $a$ and part $\M{p}$.
\end{proposition}

\end{sloppypar}

\end{document}